\documentclass[journal,draftcls,onecolumn,12pt,twoside]{IEEEtranTCOM}
\normalsize
\usepackage{amssymb,amsmath,amsfonts,amsthm}
\usepackage{graphicx}
\usepackage{url}
\usepackage{subfigure}
\newtheorem{thm}{Theorem}[section]

\newtheorem{lemma}[thm]{Lemma}

\newtheorem{remark}{Remark}
\usepackage{xcolor}
\usepackage{cite}

\ifCLASSINFOpdf
\else
\fi

\begin{document}

\title{On the Benefits of Network-level Cooperation in IoT Networks with Aggregators} 

\author{Nikolaos Pappas~\IEEEmembership{Member,~IEEE}, Ioannis Dimitriou,
         Zheng Chen~\IEEEmembership{Member,~IEEE}% <-this % stops a space
\thanks{ N. Pappas is with the Department of Science and Technology, Link\"{o}ping University, Norrk\"{o}ping SE-60174, Sweden. (e-mail: nikolaos.pappas@liu.se). I. Dimitriou is with the Department of Mathematics, University of Patras, Patra, Peloponnese, Greece.
(e-mail: idimit@math.upatras.gr). Z. Chen is with the Department of Electrical Engineering, Link\"{o}ping University, 58183 Link\"{o}ping, Sweden.  (e-mail: zheng.chen@liu.se). This work was supported in part by ELLIIT, Center for Industrial Information Technology (CENIIT), and the Swedish Foundation for Strategic Research. This work was presented in part in the 30th International Teletraffic Congress (ITC) 2018 \cite{ITC30}.}
}

\maketitle

\begin{abstract}
In this work, we consider a random access Internet of Things IoT wireless network assisted by two aggregators collecting information from two disjoint groups of sensors. The nodes and the aggregators are transmitting in a random access manner under slotted time, the aggregators perform network-level cooperation for the data collection. The aggregators are equipped with queues to store data packets that are transmitted by the network nodes and relaying them to the destination node. We characterize the throughput performance of the IoT network and we obtain the stability conditions for the queues at the aggregators. We apply the theory of boundary value problems to analyze the delay performance. Our results show that the presence of the aggregators provides significant gains in the IoT network performance, in addition, we provide useful insights regarding the scalability of the IoT network.
\end{abstract}


\IEEEpeerreviewmaketitle

\section{Introduction}
The Internet of Things (IoT) is one of the most attractive concepts in the area of information and communication technology. IoT is expected to play an important role in our daily life by supporting massive connectivity with seamless service. It involves the interconnection of many, and possibly heterogeneous objects through the Internet using different communication technologies. The objects are equipped with communications capabilities and can vary from sensors, smart objects, etc. \cite{StankovicIoT2014, KhanWirComMag2017, LenkaACCESS2018}.

The total number of IoT connections is expected to grow tremendously the next years. The application of random access protocols in the IoT networks can potentially mitigate the co-channel interference caused by massive amount of IoT devices with low signaling overhead \cite{HasanComMag2013}.  
To support the massive connectivity in future IoT networks, practical techniques are required to collect data from a large set of devices and the traditional orthogonal multi-access schemes are not efficient. The works in \cite{MalakTCOM2016, HsuICC2017, GuoTCOM2017, LopezTWC2018, SalamAccess2018, ZhangTVT2018, HattabGC2018, Rahman2018QoS} have considered data aggregation under different scenarios and have evaluated the benefits of such technique. The work in \cite{Globecom2017Latency} proposes a cloud-assisted priority-based channel access and data aggregation scheme for irregularly deployed sensor nodes to minimize the network latency and to enhance the system reliability of IoT networks. The work in \cite{ZhangTWC2019} considers an uplink IoT network where a cellular user operating as aggregator can assist IoT devices. IoT can leverage the unlicensed spectrum for machine-to-machine communications. The problem of spectrum allocation and device association in uplink two-hop IoT networks is studied in \cite{IbrahimIoT2019}.

Network-level cooperation introduced in \cite{SadekTIT2007} and \cite{RongTIT12} can be an effective alternative method for data aggregation. It is plain relaying without any physical layer considerations and it has been shown to provide large gain in terms of throughput and delay performance.  
Recently, several works have investigated relaying at the network level \cite{PappasTWC2015,PappasGlobalsip2013, pap, PappasJCN2016, DimitriouASMTA2016, NomikosSurvey2016,dim3, PappasJCN2018, DIMITRIOU2018, CristianGC2018, Dimitriou2019AdHoc}. The deployment of aggregators under network-level cooperation can improve the throughput and delay in IoT networks \cite{KimTVT2018}. However, due to the queueing delay, the stability conditions at the aggregator queues need to be considered.

In this work, we consider a random access IoT network with two aggregators which can receive and forward data in form of packets from the network nodes. 
The nodes transmit in a random access manner, which is a common assumption in IoT technologies such as LoRa \cite{LoRa2016}. The aggregators use network-level cooperation and their receiving and transmitting modules are operating in different frequency bands, thus, they have out-of-band full duplex capabilities. Furthermore, we assume multi-packet reception (MPR) capabilities at the receivers. MPR is suitable to capture the SINR (Signal-to-Interference-plus-Noise-Ratio) based model and is more realistic to model the wireless transmission. Similar assumptions to our work can be found in \cite{MalakTCOM2016, SalamAccess2018, KimTVT2018}.

\subsection{Contributions}

The contributions of our work can be summarized as follows. 
Our primary goal is to provide support for data collection in IoT networks by applying network-level cooperation. We provide the throughput analysis of an IoT network consisting of sensors that are assisted by two aggregators, from which we can gain insights on the scalability of the considered network. In addition, we study the stability conditions for the queues at the aggregators, which guarantee finite queueing delay. The main difficulty for characterizing the stability conditions lies on the interaction of the queues, which arises when the service rate of a queue depends on the state of the other. The technique of stochastic dominance is utilized to bypass this difficulty and characterize necessary and sufficient stability conditions for the queues at the aggregators.

Furthermore, we study the average delay of the packets possibly received and forwarded by the aggregators. Our system is modeled as a two-dimensional discrete time Markov chain, and we show that the generating function of the stationary joint queue length distribution can be obtained by solving a fundamental functional equation with the theory of boundary value problems. For the general MPR case, we obtain a lower and an upper bound for the average delay. In addition, we characterize in closed form expression the gap between the lower and the upper bound. These bounds as it is also seen in the numerical results appear to be tight. For the model with capture effect, i.e., a subclass of MPR model, we provide the explicit expression for the average delay. The analysis in this work can act as a framework for other research directions that involve multiple aggregators with interacting queues.

The remainder of the paper is organized as follows: Section \ref{sec:model} describes the system model and in Section \ref{sec:anal1} we present the analysis related to the network-wide throughput. In Sections \ref{anal} and \ref{sec:delaysym}, we provide the analysis for the average delay per packet. The numerical evaluation of the theoretical results is presented in Section \ref{sec:Res}, and Section \ref{sec:conclusions} concludes the paper.

\section{System Model} \label{sec:model}

\subsection{Network Layer Model}
We consider a wireless network consisting of IoT nodes/sensors/objects, which intend to communicate to a common destination/sink $D$, and two aggregators, denoted by $R_{1}$ and $R_{2}$, which can help aggregating and relaying messages from the IoT nodes to $D$. The network model is depicted in Fig. \ref{fig:model}. The nodes are located in two non-overlapping regions. There are in total $M_1$ and $M_2$ nodes within the service range of the aggregators $R_1$ and $R_2$, respectively.
Note that the transmissions in the first region cannot be overheard by the nodes in the second region and vice versa. This can be done by carefully planning for the placement of the aggregators in order to increase the coverage area without interfering with each other. \textit{In the following, we will use the terms nodes, sensors, and objects interchangeably}.

The sensors intend to transmit packets to the destination node $D$ and they are assumed to be saturated, i.e., they always have packets to transmit. In case the transmission from a sensor to $D$ fails, the aggregator can help relaying the message to $D$ and the aggregators do not generate their own traffic. We consider using \textit{network-level cooperation} at the aggregators \cite{SadekTIT2007, RongTIT12}, which means that the aggregators are cooperating as relays in a decode-and-forward manner. The packets are assumed to have equal length and the time is divided into slots, which corresponds to the transmission time of a packet. We assume that the sensors access the wireless channel randomly without any coordination among them. We consider a full multi-packet reception (MPR) channel model, which allows the receivers to successfully receive more than one packets when there are multiple transmissions in the same slot \cite{verdu}. As the result, the sink node $D$ can receive information simultaneously from the sensors and the aggregators. Note that when all nodes are transmitting, we can have in total up to $M_1+M_2$ interfering devices at the sink. This assumption is common in the literature \cite{MalakTCOM2016, KimTVT2018, ZhangTVT2018}.

We assume that different frequency bands are allocated for the transmissions from the sensors and the aggregators, thus, there is no interference between them. On the other hand, the transmission of a node creates interference to the other nodes of the same kind, i.e., there is interference between the sensors, and between the aggregators. The transmitting and receiving units of the aggregators are operating in different channels/frequency bands to avoid self-interference, which can be considered as \textit{out-of-band full duplex mode} \cite{ramakrishnan2016design}.
The aggregators are equipped with queues that store possible packets from the sensors that failed to reach the destination. The queues are assumed to have infinite capacity, thus there is no packet dropping.\footnote{In practice, the buffers have limited size, which is usually quite large. Our analysis based on the infinite buffer size assumption can capture this scenario with minor modifications and it is more general. Furthermore, our analysis can provide design guidelines for selecting the appropriate queue size in order to achieve specific performance requirements regarding throughput and delay. The arrival and service rates for the queues are defined in Section \ref{sec:anal1}.} This is a common assumption in the literature, in the IoT context \cite{KimTVT2018}.

\begin{figure}[!ht]
\centering
\includegraphics[scale=0.25]{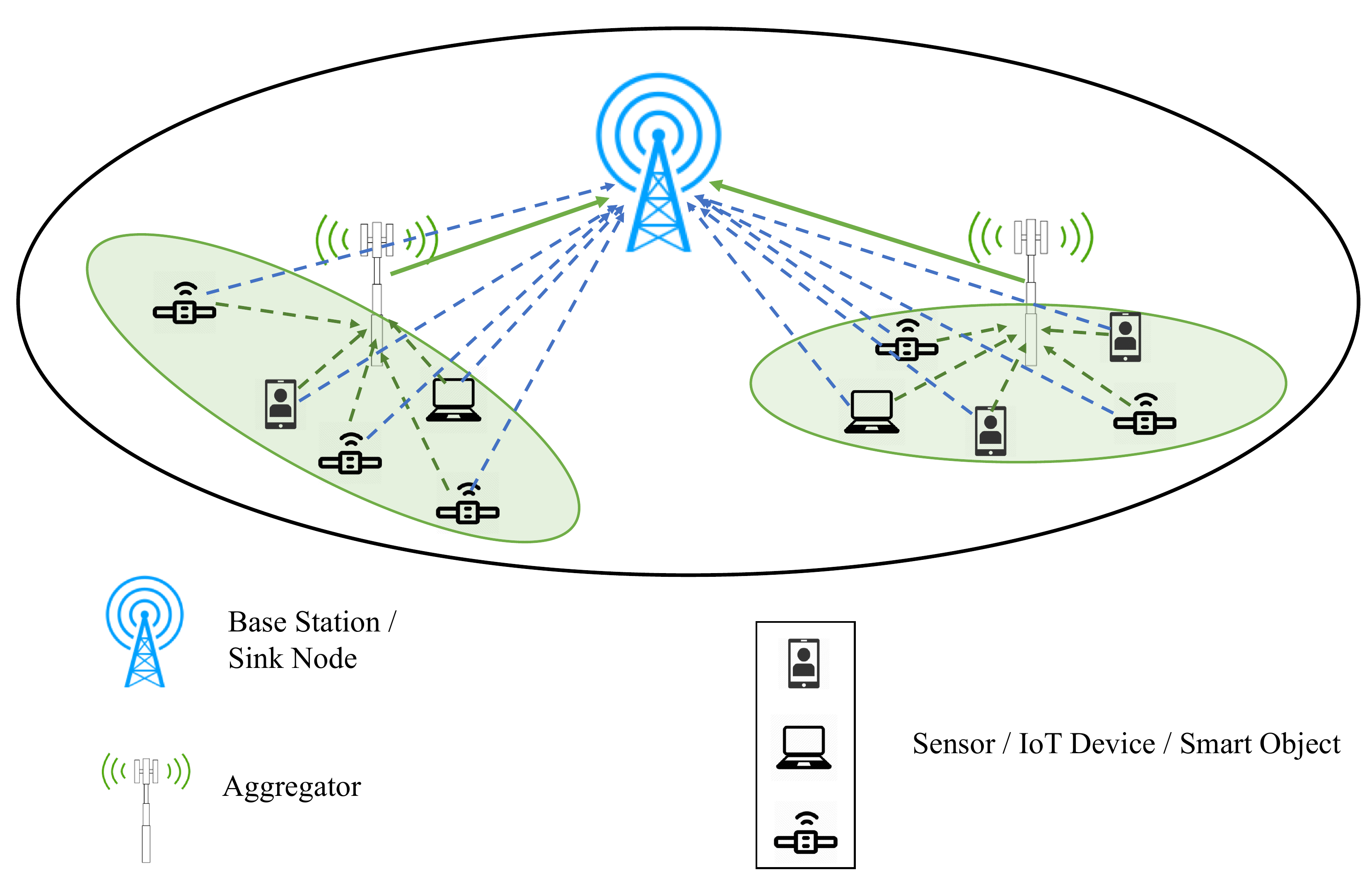}\vspace{-0.1in}
\caption{The IoT network considered in this work. The nodes are assisted by two aggregators that are equipped with queues, and they are operating in an out-of-band full duplex mode.}
\label{fig:model}
\end{figure}

\subsection{Physical Layer Model} \label{sec:PHY}
At the beginning of a timeslot, sensor nodes that belong to the coverage area of $R_i$, attempt to transmit with probability $t_i$, $i=1,2$. The aggregator $R_{i}$ will attempt to transmit a packet with probability $\alpha_{i}$, if it has a non-empty queue. Note that we assume that all the sensors in the same area transmit with the same probability. Our analysis can be easily extended to handle the general case where each node has different transmit probabilities however, the expressions will be cumbersome, and it will be difficult to extract meaningful insights.

We assume that a packet transmitted by sensor $s$ in the first coverage area is successfully received by its aggregator $R_1$ if and only if ${\rm SINR}(s,R_1)\geq \gamma$, where $\gamma$ is the SINR threshold. The wireless channel is subject to fading; let $P_{tx}(s)$ be the transmit power at sensor $s$ and $r(s,R_1)$ be the distance between the sensor and the aggregator. The received power at $R_1$, when $s$ transmits is $P_{rx}(s,R_1)=A(s,R_1)h(s,R_1)$ where $A(s,R_1)$ is a random variable representing small-scale fading. Under Rayleigh fading assumption, $A(s,R_1)$ is exponentially distributed. The received power factor $h(s,R_1)$ is given by $h(s,R_1) = P_{tx}(s)(r(s,R_1))^{-\theta}$ where $\theta$ is the path loss exponent. Then, the success probability between a sensor in the first coverage area and its aggregator is denoted by $P^{R_1}_{{i}}$, when there are $i$ sensors from the same area transmitting in a timeslot, and the expression is given by \cite{PappasTWC2015} 
\begin{equation}
\begin{aligned}
\label{eq:succprob}
P^{R_1}_{{i}}=\exp\left(-\frac{\gamma \eta_{R_1}}{v(s,R_1)h(s,R_1)}\right) \prod_{k\in \mathcal{T}\backslash \left\{s\right\}}{\left(1+\gamma\frac{v(k,R_1)h(k,R_1)}{v(s,R_1)h(s,R_1)}\right)}^{-1},
\end{aligned}
\end{equation}
where $\mathcal{T}$ is the set of transmitting sensors in the first area and $|\mathcal{T}|=i$; $v(s,R_1)$ is the parameter of the Rayleigh fading random variable, $\eta_{R_1}$ is the receiver noise power at $R_1$.

Note that transmitting sensors from the other coverage area do not create interference at the aggregator. 
However, the concurrently transmitting sensors from both areas interfere with each other at the destination $D$. Denote by $P^{1D}_{{i},{j}}$ the success probability to the sink from a sensor in coverage area $1$ when there are $i$ active transmitters from area 1 and $j$ active transmitters from area 2. Similarly we can define $P^{2D}_{{i},{j}}$.

A packet transmission from a sensor in the first area fails to reach the destination with probability $\bar{P}^{1D}_{{i},{j}}=1-P^{1D}_{{i},{j}}$, when there are $i$ active sensors in the first area and $j$ active sensors in the second area. In this case, that packet will be stored in the queue of aggregator $R_{1}$ with probability $P^{R_1}_{{i}}$. Otherwise, with probability $\bar{P}^{R_1}_{{i}}$, the aggregator fails to decode that packet and it has to be re-transmitted by the sensor in a future time slot.

Recall that if there are stored packets in the queues of the aggregator $R_{i}$, $i=1,2$, then $R_i$ transmits a packet with probability $\alpha_{i}$. If only one aggregator $R_{i}$ is active, then the packet will be successfully transmitted to $D$ with probability $p_{R_{i},\{R_{i}\}}^{D}$. If both aggregators transmit simultaneously, then with probability $p_{R_i/R_1,R_2}^{D}$ the packet from $R_i$ is successfully received by node $D$. If a transmitted packet from an aggregator fails to reach the $D$, that packet remains in the queue and will be retransmitted in a later time slot.

\section{Throughput and Stability Analysis} \label{sec:anal1}

In this section, we characterize the network throughput performance and provide the stability conditions for the queues at the aggregators.

\subsection{Throughput Analysis}
The throughput per node consists of the direct throughput from each sensor to the destination and the throughput contributed by the aggregator.
Recall that the devices that are in coverage from the first aggregator cannot cause interference at the receiver of the second aggregator. Moreover, the devices that are covered by the $i$-th aggregator, $R_i$, are transmitting with probability $t_{i}$, for $i=1,2$.

The direct throughput from a sensor in the first coverage area to the sink is given by
\begin{equation}
\begin{aligned}
T_{1,D}=\sum_{i=0}^{M_1-1}\sum_{j=0}^{M_2}{M_1-1 \choose i}{M_2 \choose j} {t_1^{i+1}\overline{t}_1^{M_1-i-1}} {t_2^{j}\overline{t}_2^{M_2-j}} P^{1D}_{{i+1},{j}}.
\end{aligned}
\end{equation}
Note that the direct throughput in this setup is equivalent to the throughput in the network without aggregators.

The contributed throughput from a sensor to the aggregator in the first coverage area is given by
\begin{equation}
\begin{aligned}
T_{1,R_1}=\sum_{i=0}^{M_1-1}\sum_{j=0}^{M_2}{M_1-1 \choose i}{M_2 \choose j} {t_1^{i+1}\overline{t}_1^{M_1-i-1}} {t_2^{j}\overline{t}_2^{M_2-j}} \overline{P}^{1D}_{{i+1},{j}}P^{R_1}_{i+1}.
\end{aligned}
\end{equation}
The total throughput seen by a sensor in the first coverage area is 
$T_1=T_{1,D}+T_{1,R_1}$. Similarly, we can obtain the throughput seen by a sensor located in the second coverage area which is assisted by the second aggregator $R_2$.
The percentage of throughput 

\begin{remark}
The term $T_{1,D}$ can also be interpreted as the probability that a transmitted packet from a sensor in the first coverage area reaches the destination directly. Furthermore, $T_{1,R_1}$ is the probability of unsuccessful transmission from a sensor to the destination while the packet is received by the aggregator. The percentage of a sensor's traffic that is being relayed is $\frac{T_{1,R_1}}{T_1}$.
\end{remark}

In addition, we need to characterize the average arrival rates at the aggregators, denoted by $\lambda_1$ and $\lambda_2$. Since we assume full MPR capability at the receivers, the $i$-th aggregator can receive up to $N_i$ packets in a timeslot.
We define $l_{k,i}$ as the probability that $k$ packets will arrive in a timeslot at the $i$-th aggregator. 
The average arrival rate at the $i$-th aggregator is given by
\begin{equation} \label{eq:lambdaju}
\lambda_{i} = \sum_{k=1}^{M_i} k l_{k,i}, \text{ }i=1,2.
\end{equation}
The probability $l_{k,1}$ where $1 \leq k \leq N_1$ is given by
\begin{equation}
l_{k,1}=\sum_{s=k}^{M_1}\sum_{m=0}^{M_2}{M_1 \choose s}{s \choose k}{M_2 \choose m} {t_1^{s}\overline{t}_1^{M_1-s}} {t_2^{m}\overline{t}_2^{M_2-m}}  \left( \overline{P}^{1D}_{{s},{m}}P^{R_1}_{s} \right)^{k} \left(P^{1D}_{{s},{m}} + \overline{P}^{1D}_{{s},{m}}\overline{P}^{R_1}_{s} \right)^{s-k}.
\end{equation}

Similarly we can obtain $l_{k,2}$.
The network-wide throughput is $T=M_1 T_1+M_2 T_2$. The previous expressions for the throughput are valid when the queues at the aggregators are stable.

\begin{remark}
Previously we considered the case where the queues are stable, but in order to obtain the throughput when the queues are not stable we needs to replace the sum of expressions $T_{i,R_i}$ with the service rate of the aggregator. In this case the network-wide throughput will be given by $T=M_1 T_{1,D}+M_2 T_{2,D}+\mu_1+\mu_2$.
\end{remark}

\subsection{Stability Analysis at the Aggregators} \label{sec:stability_conditions}
The average service rate for the aggregator $i$ is given by
\begin{equation} \label{eq:mu}
\begin{array}{l}
\mu_i=Pr(N_j \neq 0) \left[\alpha_i\bar{\alpha}_{j}p_{R_{i},\{R_{i}\}}^{D}+\alpha_i \alpha_j p_{R_i/R_i,R_j}^{D}\right] + Pr(N_j = 0) \alpha_{i} p_{R_{i},\{R_{i}\}}^{D},j=i\text{ }\mathrm{mod}2+1,
\end{array}
\end{equation}
where $N_{j}$ is the queue size at queue $j$.
The notation for the success probabilities used in (\ref{eq:mu}) is the one introduced in Section \ref{sec:model}.
The theorem below, provides the stability conditions for the considered setup.

\begin{thm}\label{thm1}
The stability conditions for the queues at the aggregators for fixed transmission probabilities $\alpha_{i}, i=1,2$, are described by the region $\mathcal{R}=\mathcal{R}_1 \cup \mathcal{R}_2$ where $\mathcal{R}_i$ is given by 
\begin{align} \label{eq:stabCond}
\mathcal{R}_i = \left\lbrace(\lambda_1,\lambda_2):\lambda_i \leq \alpha_i p_{R_{i},\{R_{i}\}}^{D} - \frac{\alpha_i \alpha_j(p_{R_{i},\{R_{i}\}}^{D}-p_{R_i/R_i,R_j}^{D})}{\alpha_j\bar{\alpha}_ip_{R_{j},\{R_{j}\}}^{D}+\alpha_i \alpha_j p_{R_j/R_i,R_j}^{D}}\lambda_j, \right. \\ \nonumber \left.  \lambda_j \leq \alpha_j\bar{\alpha}_ip_{R_{j},\{R_{j}\}}^{D}+\alpha_i \alpha_jp_{R_j/R_i,R_j}^{D} \right\rbrace, j=i\text{ }\mathrm{mod}2+1.
\end{align}
\end{thm}

\begin{proof}
The proof is given in Appendix \ref{sec:stability_proof}.
\end{proof}

\section{Preliminary Analysis}\label{anal}
Let $N_{k,n}$ be the number of packets in the buffer of aggregator $R_{k}$, $k=1,2$, at the beginning of the $n$th slot. Then, $Y_{n}=(N_{1,n},N_{2,n})$ is a discrete time Markov chain with state space $E=\{(i,j):i,j=0,1,2,...\}$. The queues of both aggregators evolve as follows:
\begin{equation}
\begin{array}{c}
N_{k,n+1}=[N_{k,n}+F_{k,n}]^{+},\,k=1,2,
\end{array}
\label{x}
\end{equation}
where $F_{k,n}$ is the the difference of the number of packets that enter the buffer of the $k$th aggregator at the beginning of slot $n$ ($F_{k,n}$ equals $0$ or $\pm 1$), and $[x]^{+}:=max(0,x)$. 

Before proceeding with the analysis, we will slightly modify the notation for the success probabilities presented in Section \ref{sec:model} in order to be more convenient for the delay analysis. If two sensors transmit a packet simultaneously, $p_{i,\{1,2\}}^{D}$ represents the probability that the packet from sensor $i$ is successfully received by node $D$ and the packet from sensor $j=i mod 2+1$ failed to be received by $D$. $p_{1,2,\{1,2\}}^{D}$ represents the probability that the packets from both nodes are successfully received by $D$. Then we have $\bar{p}_{\{1,2\}}^{D}=1-p_{1,2,\{1,2\}}^{D}-p_{1,\{1,2\}}^{D}-p_{2,\{1,2\}}^{D}$,which denotes the probability that both packets fail to be received by the node $D$. If both aggregators transmit simultaneously, then with probability $p_{R_{k},\{R_{1},R_{2}\}}^{D}$ the packet from $R_{k}$ is successfully received by node $D$, with probability $p_{R_{1},R_{2},\{R_{1},R_{2}\}}^{D}$ the packets from both aggregators are successfully received by $D$, while with probability $\bar{p}_{\{R_1,R_2\}}^{D}=1-p_{R_{1},R_{2},\{R_{1},R_{2}\}}^{D}-p_{R_{1},\{R_{1},R_{2}\}}^{D}-p_{R_{2},\{R_{1},R_{2}\}}^{D}$, both of them failed to be received by $D$.

Let $H(x,y)$ be the generating function of the joint stationary queue process,
\begin{displaymath}
\begin{array}{c}
H(x,y)=\lim_{n\to\infty}E(x^{N_{1,n}}y^{N_{2,n}}),\,|x|\leq1,|y|\leq1.
\end{array}
\end{displaymath}
Then, by exploiting (\ref{x}) we obtain after lengthy calculations 
\begin{equation}
\begin{array}{r}
R(x,y)H(x,y)=A(x,y)H(x,0)+B(x,y)H(0,y)+C(x,y)H(0,0),
\end{array}
\label{we}
\end{equation}
where, 
\begin{displaymath}
\begin{array}{rl}
R(x,y):=1-L(x,y)[1-\alpha_{1}\widehat{\alpha}_{2}(1-\frac{1}{x})-\widehat{\alpha}_{1}\alpha_{2}(1-\frac{1}{y})-\alpha_{1}\alpha_{2}p_{R_{1},R_{2},\{R_{1},R_{2}\}}^{D}(1-\frac{1}{xy})],
\end{array}
\end{displaymath}
is the \textit{kernel} of functional equation (\ref{we}) and
\begin{displaymath}
\begin{array}{rl}
A(x,y)=&L(x,y)[d_{1}(1-\frac{1}{x})+\alpha_{2}\widehat{\alpha}_{1}(1-\frac{1}{y})+\alpha_{1}\alpha_{2}p_{R_{1},R_{2},\{R_{1},R_{2}\}}^{D}(1-\frac{1}{xy})],\\
B(x,y)=&L(x,y)[d_{2}(1-\frac{1}{y})+\alpha_{1}\widehat{\alpha}_{2}(1-\frac{1}{x})+a_{1}a_{2}p_{R_{1},R_{2},\{R_{1},R_{2}\}}^{D}(1-\frac{1}{xy})],\\
C(x,y)=&L(x,y)[d_{1}(\frac{1}{x}-1)+d_{2}(\frac{1}{y}-1)+\alpha_{1}\alpha_{2}p_{R_{1},R_{2},\{R_{1},R_{2}\}}^{D}(\frac{1}{xy}-1)],\\
L(x,y)=&1-(1-x)t_{1}[\bar{t}_{2}\bar{p}_{1,\{1\}}^{D}p_{1,\{1\}}^{R_{1}}+t_{2}(\bar{p}_{\{1,2\}}^{D}+p_{2,\{1,2\}}^{D})p_{1,\{1,2\}}^{R_{1}}]+(1-y)t_{2}[\bar{t}_{1}\bar{p}_{2,\{2\}}^{D}p_{2,\{2\}}^{R_{2}}\\&+t_{1}(\bar{p}_{\{1,2\}}^{D}+p_{1,\{1,2\}}^{D})p_{1,\{1,2\}}^{R_{2}}]+(1-x)(1-y)t_{1}t_{2}\bar{p}_{\{1,2\}}^{D}p_{1,\{1,2\}}^{R_{1}}p_{2,\{1,2\}}^{R_{2}},
\end{array}
\end{displaymath}
%L(x,y)=& \bar{t}_{1}\bar{t}_{2}+t_{1}\bar{t}_{2}[p_{1,\{1\}}^{D}+\bar{p}_{1,\{1\}}^{D}(\bar{p}_{1,\{1\}}^{R_{1}}+x p_{1,\{1\}}^{R_{1}})]\\
%&+t_{2}\bar{t}_{1}[p_{2,\{2\}}^{D}+\bar{p}_{2,\{2\}}^{D}(\bar{p}_{2,\{2\}}^{R_{2}}+y p_{2,\{2\}}^{R_{2}})]\\
%&+t_{1}t_{2}[p_{1,2,\{1,2\}}^{D}+\bar{p}_{1,2,\{1,2\}}^{D}(\bar{p}_{1,\{1,2\}}^{R_{1}}+x p_{1,\{1,2\}}^{R_{1}})\\&\times(\bar{p}_{2,\{1,2\}}^{R_{2}}+y p_{2,\{1,2\}}^{R_{2}})+p_{1,\{1,2\}}^{D}(\bar{p}_{2,\{1,2\}}^{R_{2}}\\&+y p_{2,\{1,2\}}^{R_{2}})+p_{2,\{1,2\}}^{D}(\bar{p}_{1,\{1,2\}}^{R_{1}}+x p_{1,\{1,2\}}^{R_{1}})],
\begin{displaymath}
\begin{array}{rl}
\widehat{\alpha}_{k}=&\bar{\alpha}_{k}p_{R_k,\{R_k\}}^{D}+\alpha_{k}p_{R_k,\{R_1,R_2\}}^{D},\,k=1,2,\\
d_{1}=&\alpha_{1}(\widehat{\alpha}_{2}-p_{R_{1},\{R_{1}\}}^{D}),
d_{2}=\alpha_{2}(\widehat{\alpha}_{1}-p_{R_{2},\{R_{2}\}}^{D}).
\end{array}
\end{displaymath}
Clearly, for every fixed $y$ with $|y|\leq 1$, $H(x,y)$ it is regular in $x$ for $|x|<1$, and continuous in $x$ for $|x|\leq1$; similar observation hold for the variable $y$.

Some interesting relations are directly obtained using (\ref{we}). In particular, by setting in (\ref{we}) $y=1$, dividing with $x-1$, and taking the limit $x\to 1$, by using the L'Hospital rule, and vice-versa we obtain the following conservation of flow relations:
\begin{equation}
\begin{array}{rl}
\lambda_{1}=&\alpha_{1}\tilde{\alpha}_{2}[1-H(1,0)-H(0,1)+H(0,0)]+\alpha_{1}p_{R_{1},\{R_{1}\}}^{D}[H(1,0)-H(0,0)],\\
\lambda_{2}=&\alpha_{2}\tilde{\alpha}_{1}[1-H(1,0)-H(0,1)+H(0,0)]+\alpha_{2}p_{R_{2},\{R_{2}\}}^{D}[H(0,1)-H(0,0)],
\end{array}
\label{con}
\end{equation}
where, $\tilde{\alpha}_{k}:=\widehat{\alpha}_{k}+\alpha_{k}p_{R_{1},R_{2},\{R_{1},R_{2}\}}^{D}$, $k=1,2,$
\begin{displaymath}
\begin{array}{l}
\lambda_{1}:=t_{1}\bar{t}_{2}\bar{p}_{1,\{1\}}^{D}p_{1,\{1\}}^{R_{1}}+t_{1}t_{2}(\bar{p}_{1,2,\{1,2\}}^{D}+p_{2,\{1,2\}}^{D})p_{1,\{1,2\}}^{R_{1}},\\
\lambda_{2}:=t_{2}\bar{t}_{1}\bar{p}_{2,\{2\}}^{D}p_{2,\{2\}}^{R_{2}}+t_{1}t_{2}(\bar{p}_{1,2,\{1,2\}}^{D}+p_{1,\{1,2\}}^{D})p_{2,\{1,2\}}^{R_{2}}.
\end{array}
\end{displaymath}
Note that the previous equations are the same with the ones obtained in the previous section, but here we use the more convenient notation for the delay analysis regarding the success probabilities. In order to facilitate presentation we present the case of two nodes, but clearly the analysis holds for the general case of $N$ nodes, just by replacing the right parts of $\lambda_{1}$ and $\lambda_{2}$.

The expressions in (\ref{con}), equate the flow of jobs into an aggregator, with the flow of jobs out of the aggregator. Looking carefully at (\ref{con}) it is readily seen that the following analysis is distinguished in two cases:
\begin{enumerate}
\item $\frac{\alpha_{1}\tilde{\alpha}_{2}}{\alpha_{1}p_{R_{1},\{R_{1}\}}^{D}}+\frac{\alpha_{2}\tilde{\alpha}_{1}}{\alpha_{2}p_{R_{2},\{R_{2}\}}^{D}}=1$. Then, $H(0,0)=1-\frac{\lambda_{1}}{\alpha_{1}p_{R_{1},\{R_{1}\}}^{D}}-\frac{\lambda_{2}}{\alpha_{2}p_{R_{2},\{R_{2}\}}^{D}}=1-\rho$.
%\begin{equation*}
%\begin{array}{l}
%H(0,0)=1-\frac{\lambda_{1}}{\alpha_{1}p_{R_{1},\{R_{1}\}}^{D}}-\frac{\lambda_{2}}{\alpha_{2}p_{R_{2},\{R_{2}\}}^{D}}=1-\rho.
%\end{array}
%\end{equation*}
\item $\frac{\alpha_{1}\tilde{\alpha}_{2}}{\alpha_{1}p_{R_{1},\{R_{1}\}}^{D}}+\frac{\alpha_{2}\tilde{\alpha}_{1}}{\alpha_{2}p_{R_{2},\{R_{2}\}}^{D}}\neq1$. Then, (\ref{con}) yields $H(1,0)=\frac{\alpha_{1}\tilde{\alpha}_{2}(\lambda_{2}-\alpha_{2}p_{R_{2},\{R_{2}\}}^{D})-\tilde{d}_{2}(\lambda_{1}+H(0,0)\alpha_{1}p_{R_{1},\{R_{1}\}}^{D})}{\tilde{d}_{1}\tilde{d}_{2}-\alpha_{1}\alpha_{2}\tilde{\alpha}_{1}\tilde{\alpha}_{2}},$ and $H(0,1)=\frac{\alpha_{2}\tilde{\alpha}_{1}(\lambda_{1}-\alpha_{1}p_{R_{1},\{R_{1}\}}^{D})-\tilde{d}_{1}(\lambda_{2}+H(0,0)\alpha_{2}p_{R_{2},\{R_{2}\}}^{D})}{\tilde{d}_{1}\tilde{d}_{2}-\alpha_{1}\alpha_{2}\tilde{\alpha}_{1}\tilde{\alpha}_{2}},$
%\begin{displaymath}
%\begin{array}{rl}
%H(1,0)=&\frac{\alpha_{1}\tilde{\alpha}_{2}(\lambda_{2}-\alpha_{2}p_{R_{2},\{R_{2}\}}^{D})-\tilde{d}_{2}(\lambda_{1}+H(0,0)\alpha_{1}p_{R_{1},\{R_{1}\}}^{D})}{\tilde{d}_{1}\tilde{d}_{2}-\alpha_{1}\alpha_{2}\tilde{\alpha}_{1}\tilde{\alpha}_{2}},\\
%H(0,1)=&\frac{\alpha_{2}\tilde{\alpha}_{1}(\lambda_{1}-\alpha_{1}p_{R_{1},\{R_{1}\}}^{D})-\tilde{d}_{1}(\lambda_{2}+H(0,0)\alpha_{2}p_{R_{2},\{R_{2}\}}^{D})}{\tilde{d}_{1}\tilde{d}_{2}-\alpha_{1}\alpha_{2}\tilde{\alpha}_{1}\tilde{\alpha}_{2}},
%\end{array}
%\end{displaymath}
where $\tilde{d}_{1}=\alpha_{1}(\tilde{\alpha}_{2}-p_{R_{1},\{R_{1}\}}^{D})$, $\tilde{d}_{2}=a\alpha_{2}(\tilde{\alpha}_{1}-p_{R_{2},\{R_{2}\}}^{D})$.
\end{enumerate}
\subsection{Kernel analysis}
The kernel $R(x,y)$ plays a crucial role in the following analysis and here we provide some important properties. For convenience, assume in the following that $p_{R_{1},R_{2},\{R_{1},R_{2}\}}^{D}=0$. It is readily seen that
\begin{displaymath}
R(x,y)=a(x)y^{2}+b(x)y+c(x)=\widehat{a}(y)x^{2}+\widehat{b}(y)x+\widehat{c}(y),
\end{displaymath}
where, for $L=t_{1}t_{2}\bar{p}_{\{1,2\}}^{D}p_{1,\{1,2\}}^{R_{1}}p_{2,\{1,2\}}^{R_{2}}$,
\begin{displaymath}
\begin{array}{rl}
\widehat{a}(y)=&L(1-\alpha_{1}\widehat{\alpha}_{2}-\alpha_{2}\widehat{\alpha}_{1})y^{2}+y[\lambda_{1}(a_{1}\widehat{\alpha}_{2}+\alpha_{2}\widehat{\alpha}_{1}-1)\\&+L(\alpha_{1}\widehat{\alpha}_{2}+2\alpha_{2}\widehat{\alpha}_{1}-1)]-\alpha_{2}\widehat{\alpha}_{1}(\lambda_{1}+L),\\
\widehat{b}(y)=&y^{2}[\lambda_{2}(\alpha_{1}\widehat{\alpha}_{2}+\alpha_{2}\widehat{\alpha}_{1}-1)+L(2\alpha_{1}\widehat{\alpha}_{2}+\alpha_{2}\widehat{\alpha}_{1}-1)]\\
&+y[\lambda_{1}(1-2\alpha_{1}\widehat{\alpha}_{2}-\alpha_{2}\widehat{\alpha}_{1})+\lambda_{2}(1-2\widehat{\alpha}_{1}\alpha_{2}-\alpha_{1}\widehat{\alpha}_{2})\\
&+L(1-2(\widehat{\alpha}_{1}\alpha_{2}+\alpha_{1}\widehat{\alpha}_{2}))+\alpha_{1}\widehat{\alpha}_{2}+\alpha_{2}\widehat{\alpha}_{1}]\\&+\widehat{\alpha}_{1}\alpha_{2}(\lambda_{1}+\lambda_{2}+L-1),\\
\widehat{c}(y)=&\alpha_{1}\widehat{\alpha}_{2}y[\lambda_{1}-1+(\lambda_{2}+L)(1-y)],\\
a(x)=&L(1-\alpha_{1}\widehat{\alpha}_{2}-\alpha_{2}\widehat{\alpha}_{1})x^{2}+x[\lambda_{2}(\alpha_{1}\widehat{\alpha}_{2}+\alpha_{2}\widehat{\alpha}_{1}-1)\\&+L(2\alpha_{1}\widehat{\alpha}_{2}+\alpha_{2}\widehat{\alpha}_{1}-1)]-\alpha_{1}\widehat{\alpha}_{2}(\lambda_{2}+L),\\
b(x)=&x^{2}[\lambda_{1}(\alpha_{1}\widehat{\alpha}_{2}+\alpha_{2}\widehat{\alpha}_{1}-1)+L(2\alpha_{2}\widehat{\alpha}_{1}+\alpha_{1}\widehat{\alpha}_{2}-1)]\\
&+x[\lambda_{1}(1-2\alpha_{1}\widehat{\alpha}_{2}-\alpha_{2}\widehat{\alpha}_{1})+\lambda_{2}(1-2\widehat{\alpha}_{1}\alpha_{2}-\alpha_{1}\widehat{\alpha}_{2})\\
&+L(1-2(\widehat{\alpha}_{1}\alpha_{2}+\alpha_{1}\widehat{\alpha}_{2}))+\alpha_{1}\widehat{\alpha}_{2}+\alpha_{2}\widehat{\alpha}_{1}]\\&+\widehat{\alpha}_{2}\alpha_{1}(\lambda_{1}+\lambda_{2}+L-1),\\
c(x)=&\alpha_{2}\widehat{\alpha}_{1}x[\lambda_{2}-1+(\lambda_{1}+L)(1-x)].
\end{array}
\end{displaymath}
The roots of $R(x,y)=0$ are $X_{\pm}(y)=\frac{-\widehat{b}(y)\pm\sqrt{D_{y}(y)}}{2\widehat{a}(y)}$, $Y_{\pm}(x)=\frac{-b(x)\pm\sqrt{D_{x}(x)}}{2a(x)}$, where $D_{y}(y)=\widehat{b}(y)^{2}-4\widehat{a}(y)\widehat{c}(y)$, $D_{x}(x)=b(x)^{2}-4a(x)c(x)$.

\begin{lemma}\label{LEM}
For $|y|=1$, $y\neq1$, $R(x,y)=0$ has exactly one root $x=X_{0}(y)$ such that $|X_{0}(y)|<1$. When $\lambda_{1}<\alpha_{1}\widehat{\alpha}_{2}$, $X_{0}(1)=1$. Similarly, $R(x,y)=0$ has exactly one root $y=Y_{0}(x)$, such that $|Y_{0}(x)|\leq1$, for $|x|=1$.
\end{lemma}
\begin{proof}
See Appendix \ref{a0}.
\end{proof}

The lemma below provides information about the location of the branch points of the two-valued functions $Y(x)$, $X(y)$, its proof is based on algebraic arguments, thus it is omitted.

\begin{lemma}\label{lem1}
The algebraic function $Y(x)$, defined by $R(x,Y(x)) = 0$, has four real branch points $0< x_{1}<x_{2}\leq1<x_{3}<x_{4}$. Moreover, $D_{x}(x)<0$, $x\in(x_{1},x_{2})\cup(x_{3},x_{4})$ and $D_{x}(x)>0$, $x\in(-\infty,x_{1})\cup(x_{2},x_{3})\cup(x_{4},\infty)$. Similarly, $X(y)$, is defined by $R(X(y),y) = 0$, it has four real branch points $0\leq y_{1}<y_{2}\leq1<y_{3}<y_{4}$, and $D_{x}(y)<0$, $y\in(y_{1},y_{2})\cup(y_{3},y_{4})$ and $D_{x}(y)>0$, $y\in(-\infty,y_{1})\cup(y_{2},y_{3})\cup(y_{4},\infty)$. 
\end{lemma}
Let $\tilde{C}_{x}=\mathbb{C}_{x}-([x_{1},x_{2}]\cup[x_{3},x_{4}]$, $\tilde{C}_{y}=\mathbb{C}_{y}-([y_{1},y_{2}]\cup[y_{3},y_{4}]$, where $\mathbb{C}_{x}$, $\mathbb{C}_{y}$ the complex planes of $x$, $y$, respectively. In $\tilde{C}_{x}$ (resp. $\tilde{C}_{y}$), denote by $Y_{0}(x)$ (resp. $X_{0}(y)$) the root of $R(x,Y(x))=0$ (resp. $R(X(y),y)=0$) with the smallest modulus, and $Y_{1}(x)$ (resp. $X_{1}(y)$) the other one. 
Define the image contours, $\mathcal{L}=Y_{0}[\overrightarrow{\underleftarrow{x_{1},x_{2}}}]$, $\mathcal{M}=X_{0}[\overrightarrow{\underleftarrow{y_{1},y_{2}}}]$, where $[\overrightarrow{\underleftarrow{u,v}}]$ stands for the contour traversed from $u$ to $v$ along the upper edge of the slit $[u,v]$ and then back to $u$ along the lower edge of the slit. In the following lemma we provide exact characterization for the smooth and closed contours $\mathcal{L}$, $\mathcal{M}$ respectively.
\begin{lemma}\label{SQ} The algebraic function $X(y)$, $y\in[y_{1},y_{2}]$ lies on a closed contour $\mathcal{M}$, which is symmetric with respect to the real line and defined by
\begin{displaymath}
\begin{array}{l}
|x|^{2}=m(Re(x)),\,m(\delta)=\frac{\widehat{c}(\zeta(\delta))}{\widehat{a}(\zeta(\delta))},\,|x|^{2}\leq\frac{\widehat{c}(y_{2})}{\widehat{a}(y_{2})},
\end{array}
\end{displaymath}
where, $\zeta(\delta)=\frac{k_{2}(\delta)-\sqrt{k_{2}^{2}(\delta)-4k_{3}(\delta)k_{1}(\delta)}}{2k_{1}(\delta)}$, and
\begin{displaymath}
\begin{array}{rl}
k_{1}(\delta):=&a_{1}\widehat{a}_{2}(\lambda_{2}+2L(1-\delta))-(\lambda_{2}+L(1-2\delta))(1-\widehat{a}_{1}a_{2}),\\
k_{2}(\delta):=&2\delta[a_{1}\widehat{a}_{2}(\lambda_{1}+L)+\widehat{a}_{1}a_{2}(\lambda_{1}+2L)-\lambda_{1}-L]+\lambda_{1}+\lambda_{2}+L+a_{1}\widehat{a}_{2}(1-2\lambda_{1})\\
&+\widehat{a}_{1}a_{2}(1-\lambda_{1}-2(\lambda_{2}+L)),\\
k_{3}(\delta):=&-[\lambda_{2}(\widehat{a}_{1}a_{2}+a_{1}\widehat{a}_{2})+\widehat{a}_{1}a_{2}(1+(\lambda_{1}+L)(1+2\delta))].
\end{array}
\end{displaymath}
Moreover, $\beta_{0}:=\sqrt{\frac{\widehat{c}(y_{2})}{\widehat{a}(y_{2})}}$, $\beta_{1}:=-\sqrt{\frac{\widehat{c}(y_{1})}{\widehat{a}(y_{1})}}$ are the extreme right and left points of $\mathcal{M}$, respectively. Similarly, $Y(x)$, $x\in[x_{1},x_{2}]$ lies on a closed contour $\mathcal{L}$. Its exact representation is derived as for $\mathcal{M}$, and further details are omitted.
\end{lemma}
\begin{proof} See Appendix \ref{a1}.
\end{proof}
\subsection{The boundary value problems}\label{bound}
Here, we proceed with the derivation of the probability generating function of the joint stationary queue length distribution at relays. The analysis is distinguished in two cases according to the values of the parameters. 
\subsubsection{A Dirichlet boundary value problem}\label{dir}
Let $\frac{a_{1}\tilde{a}_{2}}{a_{1}p_{R_{1},\{R_{1}\}}^{D}}+\frac{a_{2}\tilde{a}_{1}}{a_{2}p_{R_{2},\{R_{2}\}}^{D}}=1$. It can be easily seen that
$A(x,y)=\frac{d_{1}}{\alpha_{1}\widehat{\alpha}_{2}}B(x,y)\Leftrightarrow A(x,y)=\frac{\alpha_{2}\widehat{\alpha}_{1}}{d_{2}}B(x,y)$. 
%\begin{displaymath}
%\begin{array}{c}
%A(x,y)=\frac{d_{1}}{\alpha_{1}\widehat{\alpha}_{2}}B(x,y)\Leftrightarrow A(x,y)=\frac{\alpha_{2}\widehat{\alpha}_{1}}{d_{2}}B(x,y).
%\end{array}
%\end{displaymath}
Therefore, for $y\in \mathcal{D}_{y}=\{y\in\mathcal{C}:|y|\leq1,|X_{0}(y)|\leq1\}$,
\begin{equation}
\begin{array}{l}
\alpha_{2}\widehat{\alpha}_{1}H(X_{0}(y),0)+d_{2}H(0,y)+\frac{\alpha_{2}\widehat{\alpha}_{1}(1-\rho)C(X_{0}(y),y)}{A(X_{0}(y),y)}=0.
\end{array}
\label{cons}
\end{equation}
Both $H(X_{0}(y),0)$, and $H(0,y)$, where $y\in \mathcal{D}_{y}-[y_{1},y_{2}]$, are analytic functions. Using analytic continuation arguments we consider (\ref{cons}) for $x\in\mathcal{M}$
\begin{equation}
\begin{array}{c}
\alpha_{2}\widehat{\alpha}_{1}H(x,0)+d_{2}H(0,Y_{0}(x))+\frac{\alpha_{2}\widehat{\alpha}_{1}(1-\rho)C(x,Y_{0}(x))}{A(x,Y_{0}(x))}=0.
\end{array}
\label{con2}
\end{equation}
By noticing that $H(0,Y_{0}(x))$ is real for $x\in\mathcal{M}$, i.e., $Y_{0}(x)\in[y_{1},y_{2}]$, we have
\begin{equation}
\begin{array}{c}
Re(iH(x,0))=w(x),\,x\in\mathcal{M},
\end{array}
\label{p1}
\end{equation}
where $w(x):=Re\left(-i\frac{C(x,Y_{0}(x))}{A(x,Y_{0}(x))}\right)(1-\rho)$. Clearly, some technical requirements are needed to be everything well defined. In particular, we have to investigate the possible poles of $H(x,0)$, $x\in S:=G_{\mathcal{M}}\cap\bar{D}_{x}^{c}$, where $G_{\mathcal{U}}$ be the interior domain bounded by $\mathcal{U}$, and $D_{x}=\{x:|x|<1\}$, $\bar{D}_{x}=\{x:|x|\leq1\}$, $\bar{D}_{x}^{c}=\{x:|x|>1\}$. This is equivalent with the investigation of the zeros of $A(x,Y_{0}(x))$, $x\in S_{x}$.
Moreover, in order to solve (\ref{p1}) we first transform the problem from $\mathcal{M}$ to the unit circle; see Appendix \ref{conf} for details. Let the conformal mapping, $z=\gamma(x):G_{\mathcal{M}}\to G_{\mathcal{C}}$, and its inverse $x=\gamma_{0}(z):G_{\mathcal{C}}\to G_{\mathcal{M}}$. 

By applying the transformation, the problem is reduced to the determination of function $\tilde{T}(z)=H(\gamma_{0}(z),0)$ regular for $z\in G_\mathcal{C}$, continuous for $z\in\mathcal{C}\cup G_\mathcal{C}$ such that, $Re(i\tilde{T}(z))=w(\gamma_{0}(z))$, $z\in\mathcal{C}$. The solution of the Dirichlet problem with boundary condition (\ref{p1}) is:
\begin{equation}
\begin{array}{c}
H(x,0)=-\frac{1-\rho}{2\pi}\int_{|t|=1}f(t)\frac{t+\gamma(x)}{t-\gamma(x)}\frac{dt}{t}+C,\,x\in\mathcal{M},
\end{array}
\label{sol1}
\end{equation}
where $f(t)=Re\left(-i\frac{C(\gamma_{0}(t),Y_{0}(\gamma_{0}(t)))}{A(\gamma_{0}(t),Y_{0}(\gamma_{0}(t)))}\right)$, $C$ a constant to be defined by setting $x=0\in G_{\mathcal{M}}$ in (\ref{sol1}) and using the fact that $H(0,0)=1-\rho$, $\gamma(0)=0$. In case $H(x,0)$ has a pole, say $\bar{x}$, we still have a Dirichlet problem for the function $(x-\bar{x})H(x,0)$. %Then,
%\begin{displaymath}
%\begin{array}{c}
%C=(1-\rho)(1+\frac{1}{2\pi}\int_{|t|=1}f(t)\frac{dt}{t}),
%\end{array}
%\end{displaymath}
Setting $t=e^{i\phi}$, $\gamma_{0}(e^{i\phi})=\rho(\psi(\phi))e^{i\psi(\phi)}$, we obtain after some algebra, $f(e^{i\phi})=\frac{d_{1}\alpha_{2}\sin(\psi(\phi))(1-Y_{0}(\gamma_{0}(e^{i\phi}))^{-1})}{\rho(\psi(\phi))k(\phi)}$,
%\begin{displaymath}
%\begin{array}{c}
%f(e^{i\phi})=\frac{d_{1}\alpha_{2}\sin(\psi(\phi))(1-Y_{0}(\gamma_{0}(e^{i\phi}))^{-1})}{\rho(\psi(\phi))k(\phi)},\end{array}
%\end{displaymath}
which is an odd function of $\phi$, and 
\begin{displaymath}
\begin{array}{c}
k(\phi)=\left[\alpha_{2}\widehat{\alpha}_{1}(1-Y_{0}^{-1}(\gamma_{0}(e^{i\phi})))+d_{1}(1-\frac{\cos(\psi(\phi))}{\rho(\psi(\phi))})\right]^{2}+\left(d_{1}\frac{\sin(\psi(\phi))}{\rho(\psi(\phi))}\right)^{2}.
\end{array}
\end{displaymath}
Thus, $C=1-\rho$. Substituting in (\ref{sol1}) we obtain after simple calculations an integral representation of $H(x,0)$ on a real interval for $x\in G_{\mathcal{M}}$, i.e.,
\begin{equation}
\begin{array}{c}
H(x,0)=(1-\rho)\left\lbrace 1+\frac{2\gamma(x)i}{\pi}\int_{0}^{\pi}\frac{f(e^{i\phi})\sin(\phi)d\phi}{1-2\gamma(x)\cos(\phi)-\gamma(x)^{2}}\right\rbrace.
\end{array}
\label{soll}
\end{equation}

Similarly, we can determine $H(0,y)$ by solving another Dirichlet boundary value problem on the closed contour $\mathcal{L}$. Then, using the fundamental functional equation (\ref{we}), we uniquely obtain $H(x,y)$.
\subsubsection{A homogeneous Riemann-Hilbert boundary value problem}\label{rh}
In case $\frac{a_{1}\tilde{a}_{2}}{a_{1}^{*}p_{R_{1},\{R_{1}\}}^{D*}}+\frac{a_{2}\tilde{a}_{1}}{a_{2}^{*}p_{R_{2},\{R_{2}\}}^{D*}}\neq1$,
consider the following transformation:
\begin{displaymath}
%\begin{array}{rl}
G(x):=H(x,0)+\frac{\alpha_{1}p_{R_{1},\{R_{1}\}}^{D} d_{2}H(0,0)}{d_{1}d_{2}-\alpha_{1}\widehat{\alpha}_{2}\alpha_{2}\widehat{\alpha}_{1}},
L(y):=H(0,y)+\frac{\alpha_{2}p_{R_{2},\{R_{2}\}}^{D} d_{1}H(0,0)}{d_{1}d_{2}-\alpha_{1}\widehat{\alpha}_{2}\alpha_{2}\widehat{\alpha}_{1}}.
%\end{array}
\end{displaymath}
Then, for $y\in \mathcal{D}_{y}$,
\begin{equation}
\begin{array}{c}
A(X_{0}(y),y)G(X_{0}(y))=-B(X_{0}(y),y)L(y).
\end{array}
\label{po31}
\end{equation}
Using similar arguments as in previous subsection, we have for $x\in\mathcal{M}$
\begin{equation}
\begin{array}{c}
A(x,Y_{0}(x))G(x)=-B(x,Y_{0}(x))L(Y_{0}(x)).
\end{array}
\label{za1}
\end{equation}
Clearly, $G(x)$ is holomorphic in $D_{x}$, continuous in $\bar{D}_{x}$, but it might have poles in $S_{x}$, based on the values of the system parameters. These poles, if exist, coincide with the zeros of $A(x,Y_{0}(x))$ in $S_{x}$. For $y\in[y_{1},y_{2}]$, let $X_{0}(y)=x\in\mathcal{M}$ and realize that $Y_{0}(X_{0}(y))=y$ so that $y=Y_{0}(x)$. Taking into account the possible poles of $G(x)$, and noticing that $L(Y_{0}(x))$ is real for $x\in\mathcal{M}$, since $Y_{0}(x)\in[y_{1},y_{2}]$, we have
\begin{equation}
\begin{array}{c}
Re[iU(x)\tilde{G}(x)]=0,\,x\in\mathcal{M},\\
U(x)=\frac{A(x,Y_{0}(x))}{(x-\bar{x})^{r}B(x,Y_{0}(x))},\,\tilde{G}(x)=(x-\bar{x})^{r}G(x),
\end{array}
\label{df3}
\end{equation}
where $r=0,1$, whether $\bar{x}$ is zero or not of $A(x,Y_{0}(x))$ in $S_{x}$. Thus, $\tilde{G}(x)$ is regular for $x\in G_{\mathcal{M}}$, continuous for $x\in\mathcal{M}\cup G_{\mathcal{M}}$, and $U(x)$ is a non-vanishing function on $\mathcal{M}$. By conformally transform the problem (\ref{df3}) from $\mathcal{M}$ to the unit circle, using $z=\gamma(x):G_{\mathcal{M}}\to G_{\mathcal{C}}$, and its inverse given by $x=\gamma_{0}(z):G_{\mathcal{C}}\to G_{\mathcal{M}}$, the problem in (\ref{df3}) is reduced to the following: find a function $F(z):=\tilde{G}(\gamma_{0}(z))$, regular in $G_{\mathcal{C}}$, continuous in $G_{\mathcal{C}}\cup\mathcal{C}$ such that, $Re[iU(\gamma_{0}(z))F(z)]=0,\,z\in\mathcal{C}$. 

We now need additional information in order to derive a solution for the above problem, i.e., we must determine the index $\chi=\frac{-1}{\pi}[arg\{U(x)\}]_{x\in \mathcal{M}}$, where $[arg\{U(x)\}]_{x\in \mathcal{M}}$, denotes the variation of the argument of the function $U(x)$ as $x$ moves along the closed contour $\mathcal{M}$ in the positive direction, provided that $U(x)\neq0$, $x\in\mathcal{M}$. Following the methodology in \cite{fay}, the value of the index is closely related to the stability conditions. The following lemma provides the necessary information.
\begin{lemma}\begin{enumerate}
\item If $\lambda_{2}<\alpha_{2}\widehat{\alpha}_{1}$, then $\chi=0$ is equivalent to
\begin{displaymath}
\begin{array}{l}
\frac{d A(x,Y_{0}(x))}{dx}|_{x=1}<0\Leftrightarrow\lambda_{1}<\alpha_{1}p_{R_{1},\{R_{1}\}}^{D}+\frac{d_{1}\lambda_{2}}{\alpha_{2}\widehat{\alpha}_{1}},\\ \frac{d B(X_{0}(y),y)}{dy}|_{y=1}<0\Leftrightarrow\lambda_{2}<\alpha_{2}p_{R_{2},\{R_{2}\}}^{D}+\frac{d_{2}\lambda_{1}}{\alpha_{1}\widehat{\alpha}_{2}}.
\end{array}
\end{displaymath}
\item If $\lambda_{2}\geq \alpha_{2}\widehat{\alpha}_{1}$, $\chi=0$ is equivalent to 
\begin{displaymath}
\begin{array}{c}
\frac{d B(X_{0}(y),y)}{dy}|_{y=1}<0\Leftrightarrow \lambda_{2}<\alpha_{2}p_{R_{2},\{R_{2}\}}^{D}+\frac{d_{2}\lambda_{1}}{\alpha_{1}\widehat{\alpha}_{2}}.
\end{array}
\end{displaymath}
\end{enumerate}
\end{lemma}
Thus, under stability conditions the problem defined in (\ref{df3}) has a unique solution for $x\in G_{\mathcal{M}}$ given by,
\begin{equation}
%\begin{array}{rl}
H(x,0)=K(x-\bar{x})^{-r}\exp\left\lbrace\frac{1}{2i\pi}\int_{|t|=1}\frac{\log J(t)}{t-\gamma(x)}dt \right\rbrace-\frac{\alpha_{1}p_{R_{1},\{R_{1}\}}^{D} d_{2}H(0,0)}{d_{1}d_{2}-\alpha_{1}\widehat{\alpha}_{2}\alpha_{2}\widehat{\alpha}_{1}},
%\end{array}
\label{sool1}
\end{equation}
where $K$ is a constant and $J(t)=\frac{\overline{U_{1}(t)}}{U_{1}(t)}$, $U_{1}(t)=U(\gamma_{0}(t))$, $|t|=1$. Setting $x=0$ in (\ref{sool1}) we derive a relation between $K$ and $H(0,0)$. Then, for $x=1\in G_{\mathcal{M}}$, and using the first in (\ref{con}) we can obtain $K$ and $H(0,0)$. Substituting back in (\ref{sool1}) we obtain for $x\in G_{\mathcal{M}}$,
\begin{equation}
\begin{array}{l}
H(x,0)=\frac{\lambda_{1}d_{2}+\alpha_{1}\widehat{\alpha}_{2}(\bar{t}_{1}\bar{t}_{2}\alpha_{2}p_{R_{2},\{R_{2}\}}^{D}-\lambda_{2})}{(\bar{x}-1)^{r}\bar{t}_{1}\bar{t}_{2}(\alpha_{1}\widehat{\alpha}_{2}\alpha_{2}\widehat{\alpha}_{1}-d_{1}d_{2})} \times\\\left((\bar{x}-x)^{r}\exp\left[\frac{\gamma(x)-\gamma(1)}{2i\pi}\int_{|t|=1}\frac{\log J(t)}{(t-\gamma(x))(t-\gamma(1))}dt\right]+\frac{\alpha_{1}p_{R_{1},\{R_{1}\}}^{D} d_{2}\bar{x}^{r}}{\alpha_{1}\widehat{\alpha}_{2}\alpha_{2}p_{R_{2},\{R_{2}\}}^{D}}\exp\left[\frac{-\gamma(1)}{2i\pi}\int_{|t|=1}\frac{\log J(t)}{t(t-\gamma(1))}dt \right]\right).
\end{array}
\label{fin}
\end{equation}
The other unknown function $H(0,y)$ is determined similarly, by solving another Riemann-Hilbert boundary value problem on the closed contour $\mathcal{L}$. Then, using the fundamental functional equation (\ref{we}) we uniquely obtain $H(x,y)$. After that, we can derive several performance metrics. In particular, using Little's law, the expected delay at each aggregator is
\begin{displaymath}
\begin{array}{c}
D_{1}=\frac{\lambda_{1}+d_{1}H_{1}(1,0)}{\lambda_{1}a_{1}\widehat{a}_{2}},\,D_{2}=\frac{\lambda_{2}+d_{2}H_{2}(0,1)}{\lambda_{2}a_{2}\widehat{a}_{1}},
\end{array}
\end{displaymath}
where $H_{1}(1,0):=\frac{d}{dx}H(x,0)|_{x=1}$, and $H_{2}(0,1):=\frac{d}{dy}H(0,y)|_{y=1}$.

\section{Explicit bounds for the mean delay in the symmetrical system} \label{sec:delaysym}
In the following we consider the symmetrical model and obtain explicit bounds for the average delay. In particular, for $k=1,2$ let $\alpha_{k}=\alpha$, $t_{k}=t$, $p_{k,\{k\}}^{R_{k}}=s_{1}^{R}$, $p_{k,\{k\}}^{D}=s_{1}^{D}$, $p_{k,\{1,2\}}^{R_{k}}=s_{2}^{R}$, $p_{k,\{1,2\}}^{D}=s_{2}^{D}$, $p_{1,2,\{1,2\}}^{D}=s_{0}^{D}$, $p_{R_k,\{R_k\}}^{D}=r_{1}^{D}$, $p_{R_k,\{R_1,R_{2}\}}^{D}=r_{2}^{D}$, $p_{R_1,R_2,\{R_1,R_2\}}^{D}=r_{0}^{D}$, and as a result $d_{k}=d:=\alpha^{2}(r_{2}^{D}-r_{1}^{D})$. Note that due to the symmetry, $H(0,1)=H(1,0)$. Using the fact that $H(1,1)=1$, and setting in (\ref{we}) $y=1$, dividing with $x-1$, and taking the limit $x\to 1$, by using the L'Hospital rule,
\begin{equation}
\begin{array}{rl}
\alpha[\widehat{\alpha}+\alpha r_{0}^{D}]-\lambda=&(d+\alpha\widehat{\alpha}+2\alpha^{2}r_{0}^{D})H(1,0)-(d+\alpha^{2}r_{0}^{D})H(0,0),
\end{array}
\label{xb}
\end{equation}
where now due to the symmetry $\lambda_{k}=\lambda$, $k=1,2,$ $\widehat{\alpha}=\bar{\alpha}r_{1}^{D}+\alpha r_{2}^{D}$, and $\lambda=t^{2}s_{2}^{R}(\bar{s}_{0}^{D}+s_{2}^{D})+t\bar{t}\bar{s}_{1}^{D}s_{1}^{R}$.
%\begin{displaymath}
%\lambda=t^{2}s_{2}^{R}(\bar{s}_{0}^{D}+s_{2}^{D})+t\bar{t}\bar{s}_{1}^{D}s_{1}^{R}.
%\end{displaymath}
Note that $\alpha[\widehat{\alpha}+\alpha r_{0}^{D}]>\lambda$ due to the stability condition.

Denote by $H_{1}(x,y)$, $H_{2}(x,y)$ the derivatives of $H(x,y)$ with respect to $x$, $y$. Due to the symmetry, $H_{1}(1,1)=H_{2}(1,1)$, $H_{1}(1,0)=H_{2}(0,1)$. Denote by $M_{k}=H_{1}(1,1)$, $k=1,2$. Differentiating (\ref{we}) with respect to $x$, setting $(x,y)=(1,1)$, and using (\ref{xb}) we obtain 
\begin{equation}
\begin{array}{l}
M_{1}=\frac{\lambda\bar{\lambda}+(d+\alpha^{2}r_{0}^{D})H_{1}(1,0)}{\alpha[\widehat{\alpha}+\alpha r_{0}^{D}]-\lambda},
\end{array}
\label{d1}
\end{equation}
where $\bar{\lambda}=1-\lambda$. Now set $x=y$ in (\ref{we}) we obtain
\begin{equation}
\begin{array}{l}
2(\alpha[\widehat{\alpha}+\alpha r_{0}^{D}]-\lambda)\frac{d}{dx}H(x,x)|_{x=1}=2(1+2\lambda)(\alpha[\widehat{\alpha}+\alpha r_{0}^{D}]-\lambda)+\alpha^{2}r_{0}^{D}P(N_{1}>0,N_{2}>0)\\
+2H_{1}(1,0)(d+\alpha\widehat{\alpha}+2\alpha^{2}r_{0}^{D})-2\alpha(\widehat{\alpha}+\alpha r_{0}^{D})
+t^{2}\bar{s}_{0}^{D}(s_{2}^{R})^{2}+4\lambda(1-\alpha\widehat{\alpha}-\alpha^{2}r_{0}^{D}).
\end{array}
\label{d2}
\end{equation}
Using (\ref{d1}), (\ref{d2}), and realizing that due to the symmetry $\frac{d}{dx}H(x,x)|_{x=1}=2M_{1}$, we obtain
\begin{equation}
\begin{array}{l}
M_{1}=M_{2}=\frac{\lambda\bar{\lambda}(d+\alpha\widehat{a}+2\alpha^{2}r_{0}^{D})}{2\alpha(\alpha[\widehat{\alpha}+\alpha r_{0}^{D}]-\lambda)}
-\frac{(d+\alpha^{2}r_{0}^{D})[2\lambda(1+2\bar{\lambda})+t^{2}\bar{s}_{0}^{D}(s_{2}^{R})^{2}+\alpha^{2}r_{0}^{D}P(N_{1}>0,N_{2}>0)]}{2 \alpha(\alpha[\widehat{\alpha}+\alpha r_{0}^{D}]-\lambda)}.
\end{array}
\end{equation}
Using Little's law, the average delay at each aggregator is
\begin{equation}
\begin{array}{c}
D_{1}=D_{2}:=S-\phi,
\end{array}
\label{del}
\end{equation}
where $S=\frac{\lambda\bar{\lambda}(d+\alpha\widehat{a}+2\alpha^{2}r_{0}^{D})-(d+\alpha^{2}r_{0}^{D})[2\lambda(1+2\bar{\lambda})+t^{2}\bar{s}_{0}^{D}(s_{2}^{R})^{2}]}{2\lambda \alpha(\alpha[\widehat{\alpha}+\alpha r_{0}^{D}]-\lambda)}$ and $\phi=\frac{(d+\alpha^{2}r_{0}^{D})\alpha^{2}r_{0}^{D}P(N_{1}>0,N_{2}>0)}{2\lambda \alpha(\alpha[\widehat{\alpha}+\alpha r_{0}^{D}]-\lambda)}$.

Note that in case the destination can successfully receive at most one packet from the aggregators even if both of them transmit, i.e., $r_{0}^{D}=0$, the exact average queueing delay in a node is given by (\ref{del}) for $\phi=0$. Moreover, if $r_{0}^{D}=r_{2}^{D}=0$, then (\ref{del}) provides also the exact average delay for the case of a \textit{collision} channel. For the case of the general MPR model where $r_{0}^{D}\neq 0$, we can easily obtain bounds for the average delay at aggregators based on the sign of $\phi$. Since $Pr(N_1 > 0, N_2 > 0) > 0$, the sign of $\phi$ coincides with the
sign of $d+\alpha^{2}r_{0}^{D}$. To proceed, we distinguish the analysis in the following two cases:
\begin{enumerate}
\item If $d+\alpha^{2}r_{0}^{D}<0$, then $0\leq\phi\leq -\frac{\alpha^{2}r_{0}^{D}(d+\alpha^{2}r_{0}^{D})}{2\lambda \alpha(\alpha[\widehat{\alpha}+\alpha r_{0}^{D}]-\lambda)}$.
Thus, the upper and lower delay bound, say $D_{1}^{low}$, $D_{1}^{up}$ respectively are,
\begin{displaymath}
\begin{array}{lr}
D_{1}^{low}=S,&D_{1}^{up}=D_{1}^{low}-\frac{\alpha^{2}r_{0}^{D}(d+\alpha^{2}r_{0}^{D})}{2\lambda \alpha(\alpha[\widehat{\alpha}+\alpha r_{0}^{D}]-\lambda)}.
\end{array}
\end{displaymath}
\item If $d+\alpha^{2}r_{0}^{D}>0$, then $-\frac{\alpha^{2}r_{0}^{D}(d+\alpha^{2}r_{0}^{D})}{2\lambda \alpha(\alpha[\widehat{\alpha}+\alpha r_{0}^{D}]-\lambda)}\leq\phi\leq 0$. 
In such a case,
\begin{displaymath}
\begin{array}{rl}
D_{1}^{up}=S,&
D_{1}^{low}=D_{1}^{up}-\frac{\alpha^{2}r_{0}^{D}(d+\alpha^{2}r_{0}^{D})}{2\lambda \alpha(\alpha[\widehat{\alpha}+\alpha r_{0}^{D}]-\lambda)}.
\end{array}
\end{displaymath}
\end{enumerate}

The following remark shows that the gap between the upper and the lower bound of the average delay has a closed form expression.

\begin{remark}
We observe that the gap between the lower and the upper average delay is given by $D_{1}^{up} - D_{1}^{low} =\left|\frac{\alpha^{2}r_{0}^{D}(d+\alpha^{2}r_{0}^{D})}{2\lambda \alpha(\alpha[\widehat{\alpha}+\alpha r_{0}^{D}]-\lambda)} \right|$. Note that when the arrival rate $\lambda$ increases, then the gap tends to zero. Furthermore, when $\lambda\to 0$, then the average delay is very close to zero thus, the lower bound on the delay is tight.
\end{remark}
\begin{remark}
In the case of a collision channel model assumption, which is equivalent to $r_{0}^{D}=r_{2}^{D}=0$, we have that $D_{1}^{up} = D_{1}^{low}$. This is also true for the case where $r_{0}^{D}=0$, which means that the destination can successfully receive at most one packet from the aggregators when both of them transmit. This is known as the capture effect which is common in LoRaWAN \cite{LoRaCapture2017}.
\end{remark}

\section{Numerical Results} \label{sec:Res}
In this section, we provide numerical results to evaluate the presented theoretical performance analysis. For exposition convenience, we consider the case where all sensors have the same link characteristics and transmission probabilities.

We consider the following network setup, the distance between the sensors and the sink is $130m$, the distance between the sensors and the aggregator is $60m$, and the distance between the aggregators and the destination is $80m$. The path loss exponent is assumed $4$, and the transmission power for each sensor is $1mW$ and for each aggregator is $10mW$. We assume Rayleigh block fading channel model.
The values for the SINR threshold are considered in this section are $0.2$, $0.5$, $1.2$, and $2$. In addition, regarding the transmission probabilities, for the sensors we consider the values $t=0.1$, $t=0.2$, and for the aggregators $\alpha=0.8$.

\subsection{Stability Region}

Here we present the closure of the stability region for the queues at the aggregators for all the possible random access probabilities $\alpha$. We consider two cases for the SINR threshold, the case where SINR$=0.5$, which is depicted in Fig. \ref{fig:stability_low}. \textit{We observe that this region is a convex set, thus, the system performs better than time sharing schemes and this is an indication of strong MPR capabilities at the destination}. The stability region for SINR$=1.2$ is depicted in Fig. \ref{fig:stability_high}. \textit{In this case, a time sharing scheme has higher performance since the destination has weak MPR capabilities}.

\begin{figure}[!ht]
	\centering
	\subfigure[SINR$=0.5$]{
		\includegraphics[scale=0.38]{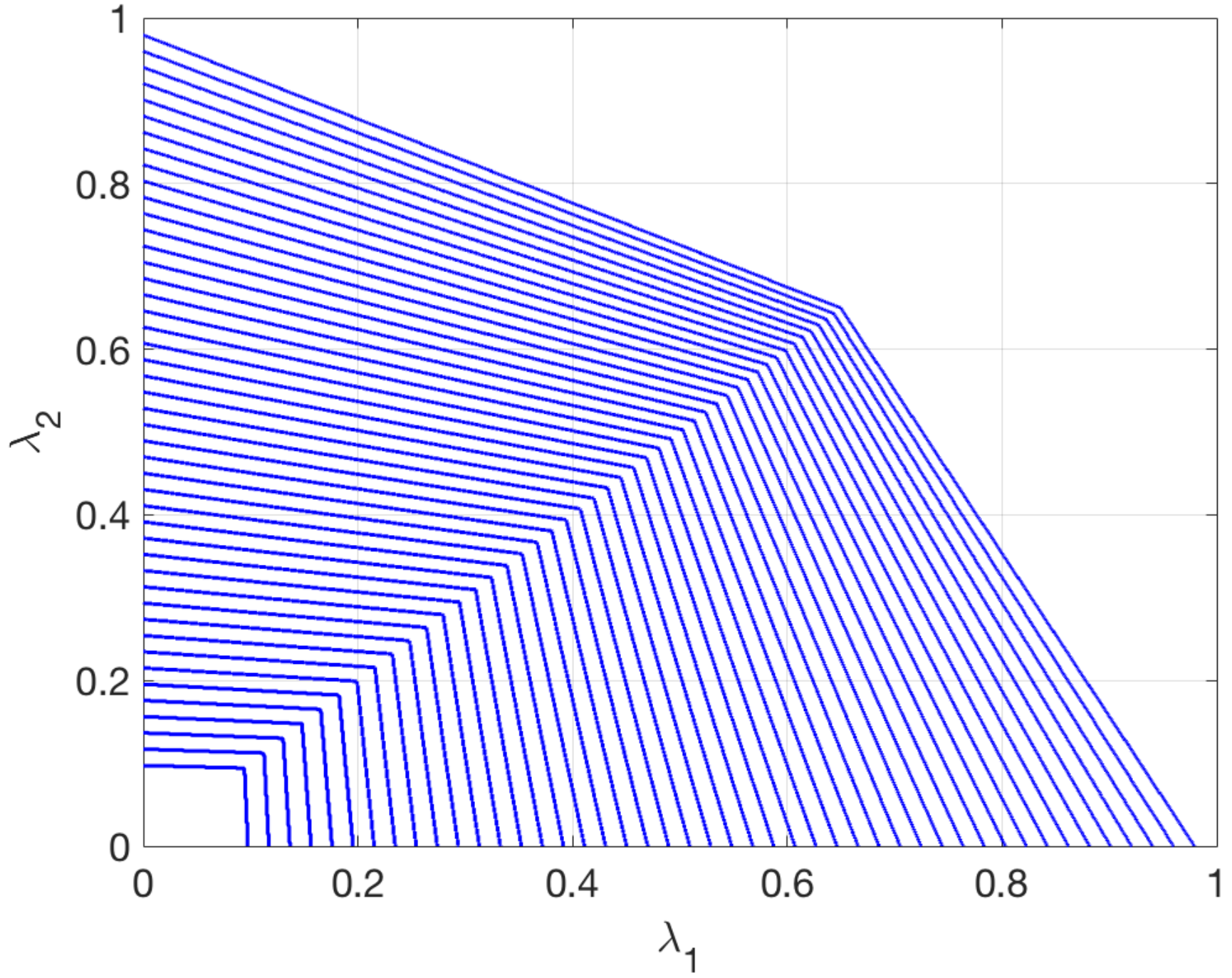}
		\label{fig:stability_low}
	}
	\subfigure[SINR$=1.2$]{
		\includegraphics[scale=0.38]{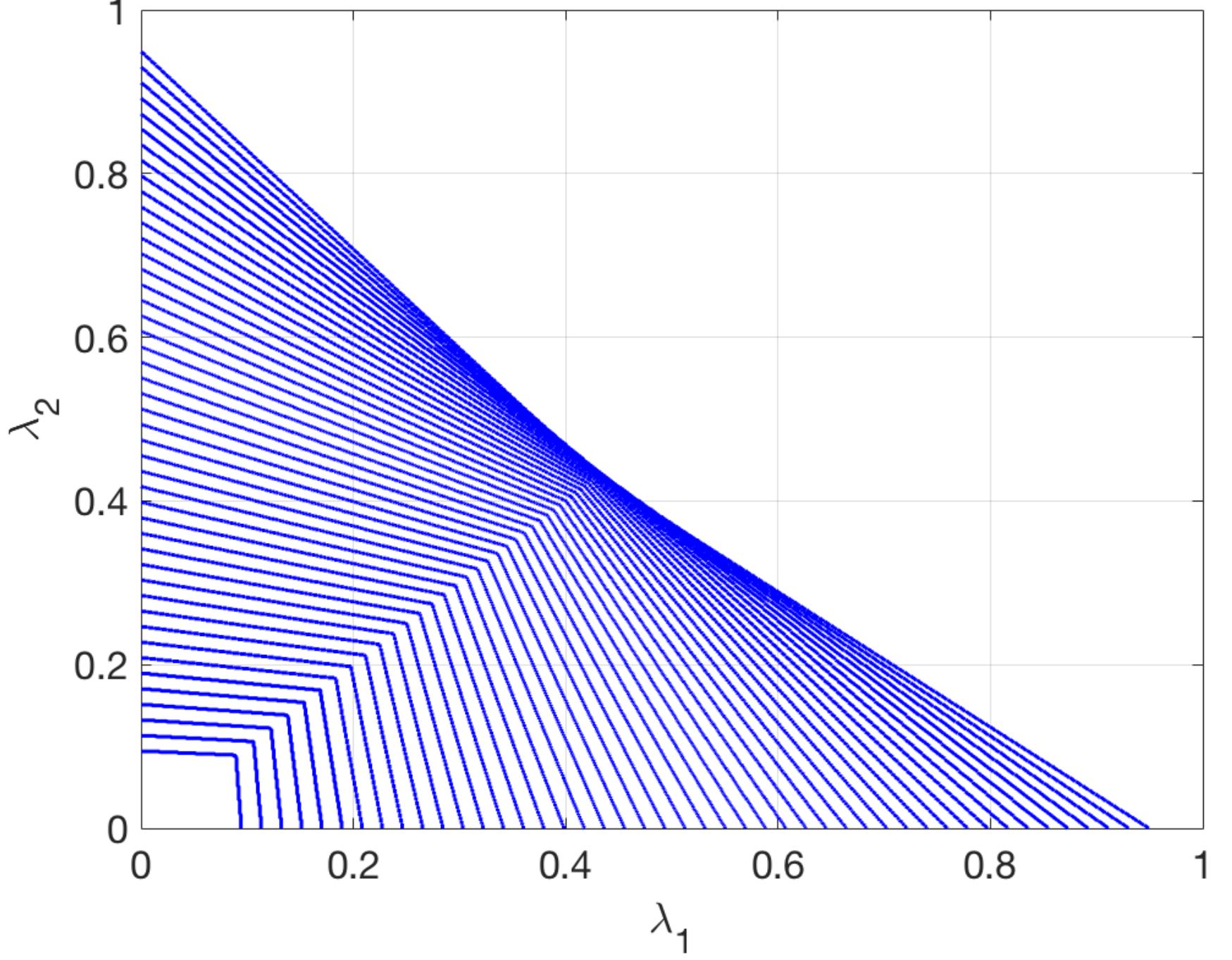}
		\label{fig:stability_high}
	}\vspace{-0.1in}
	\caption{The stability region for the queues at the aggregators.}
\end{figure}

\subsection{Network Throughput}

In this subsection, we present the performance in terms of throughput per sensor and the aggregate throughput of the considered system. We consider four cases for the SNR/SINR threshold,
two cases for the transmission probabilities of the sensors, $t=0.1$ and $t=0.2$, and for the aggregators $\alpha=0.8$. We also plot the throughput of the IoT network without the presence of the aggregators. The case where SINR$ < 1$, is depicted in Fig. \ref{fig:thlow}, in this case, we have strong MPR capabilities at the receivers, thus it is more likely to have more concurrent successful transmissions. Furthermore, we observe that the performance of the IoT network without the assistance of the aggregators can be sufficient for the case with low number of sensors.

\begin{figure}[!ht]
	\centering
	\subfigure[Throughput per sensor]{
		\includegraphics[scale=0.45]{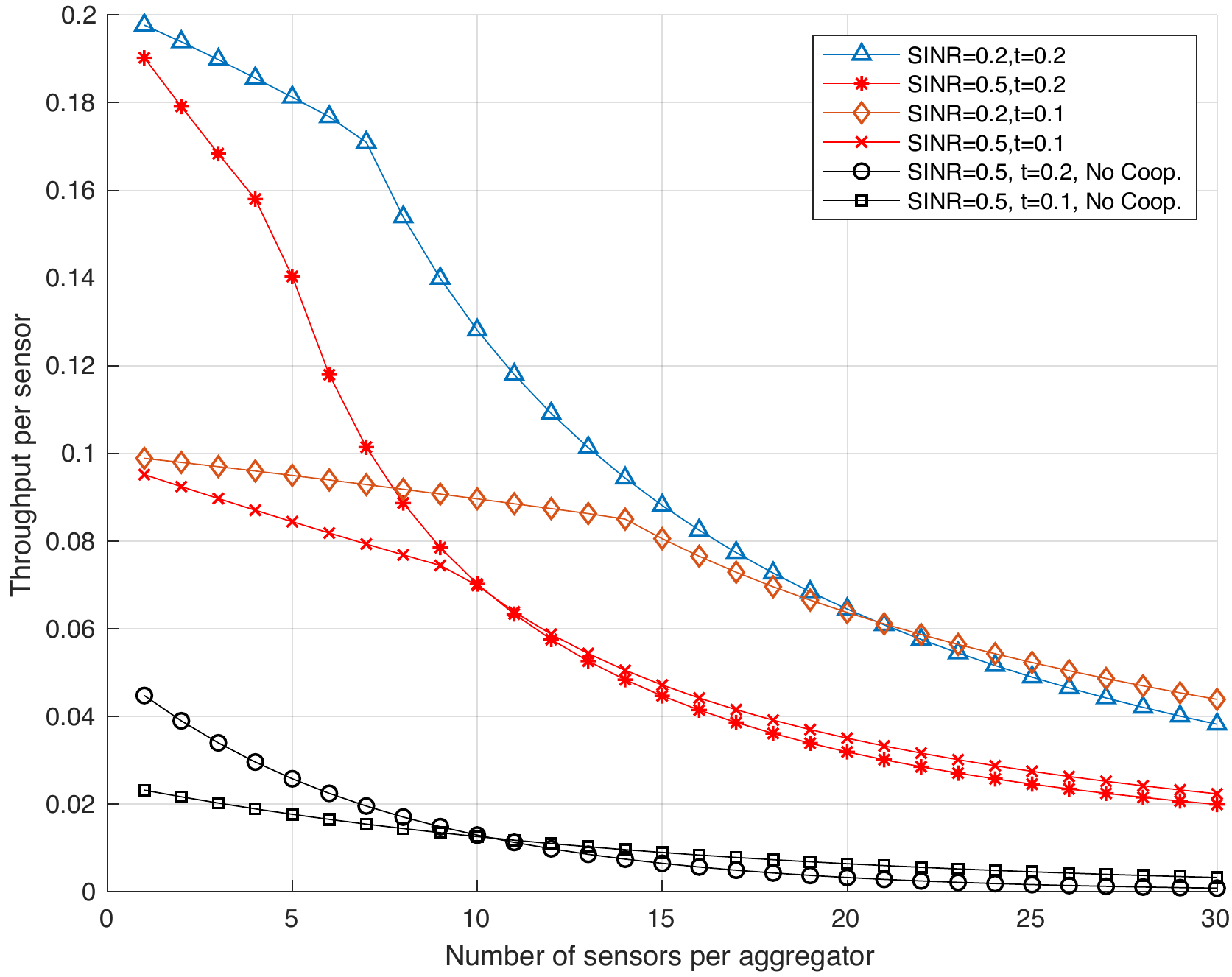}
		\label{fig:thrlow}
	}
	\subfigure[Aggregate Throughput]{
		\includegraphics[scale=0.45]{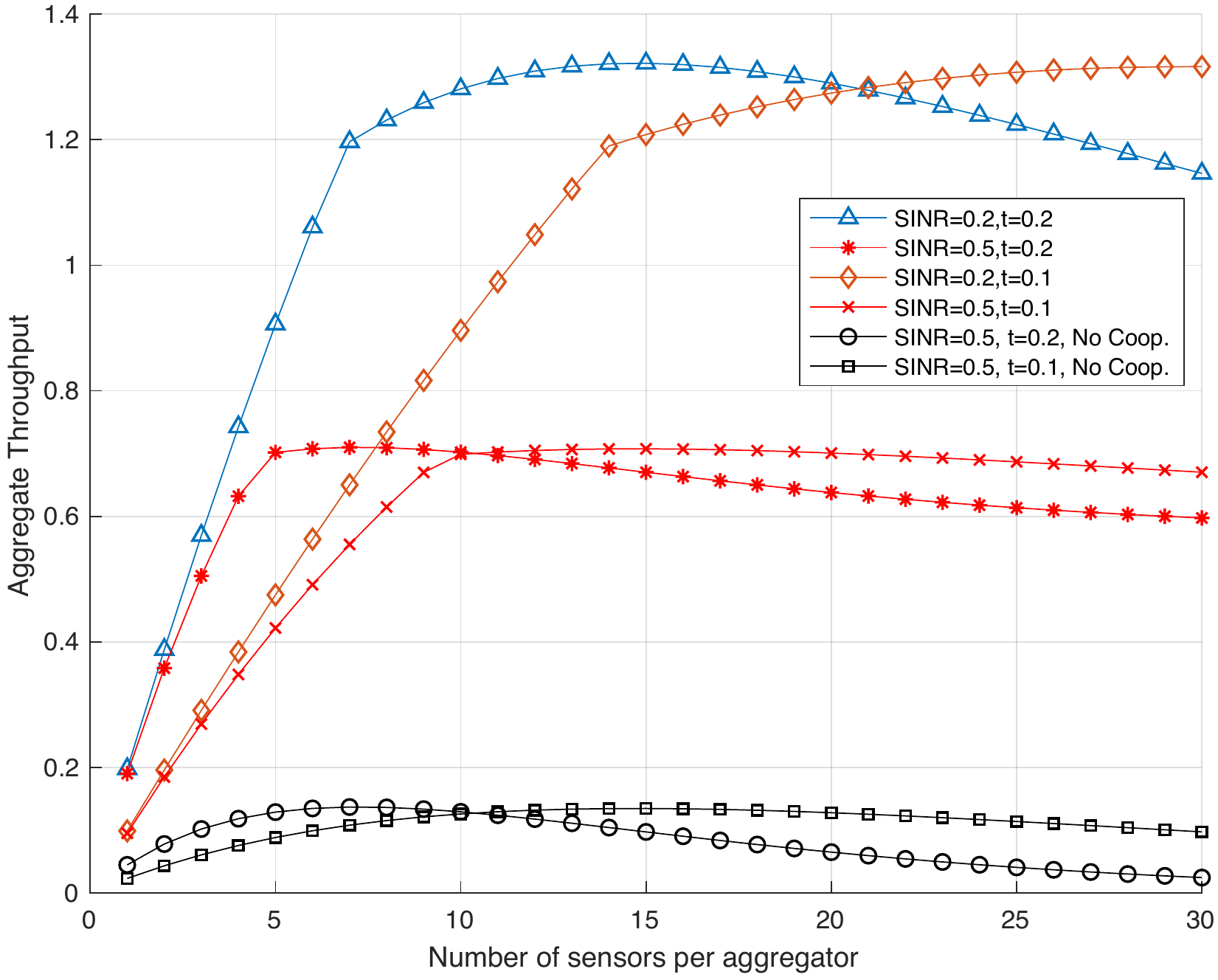}
		\label{fig:athrlow}
	}\vspace{-0.1in}
	\caption{Throughput performance of the IoT network for SINR$ < 1$.} \label{fig:thlow}
\end{figure}

The case where SINR$ > 1$, is depicted in Fig. \ref{fig:thhigh}. In this case, the benefits of the aggregators deployment are prominent. The throughput of the IoT network without the aggregators is almost zero. \textit{The presence of the aggregators that deploy network-level cooperation provides significant gains in the throughput performance of the network.}
\begin{figure}[!ht]
	\centering
	\subfigure[Throughput per sensor]{
		\includegraphics[scale=0.45]{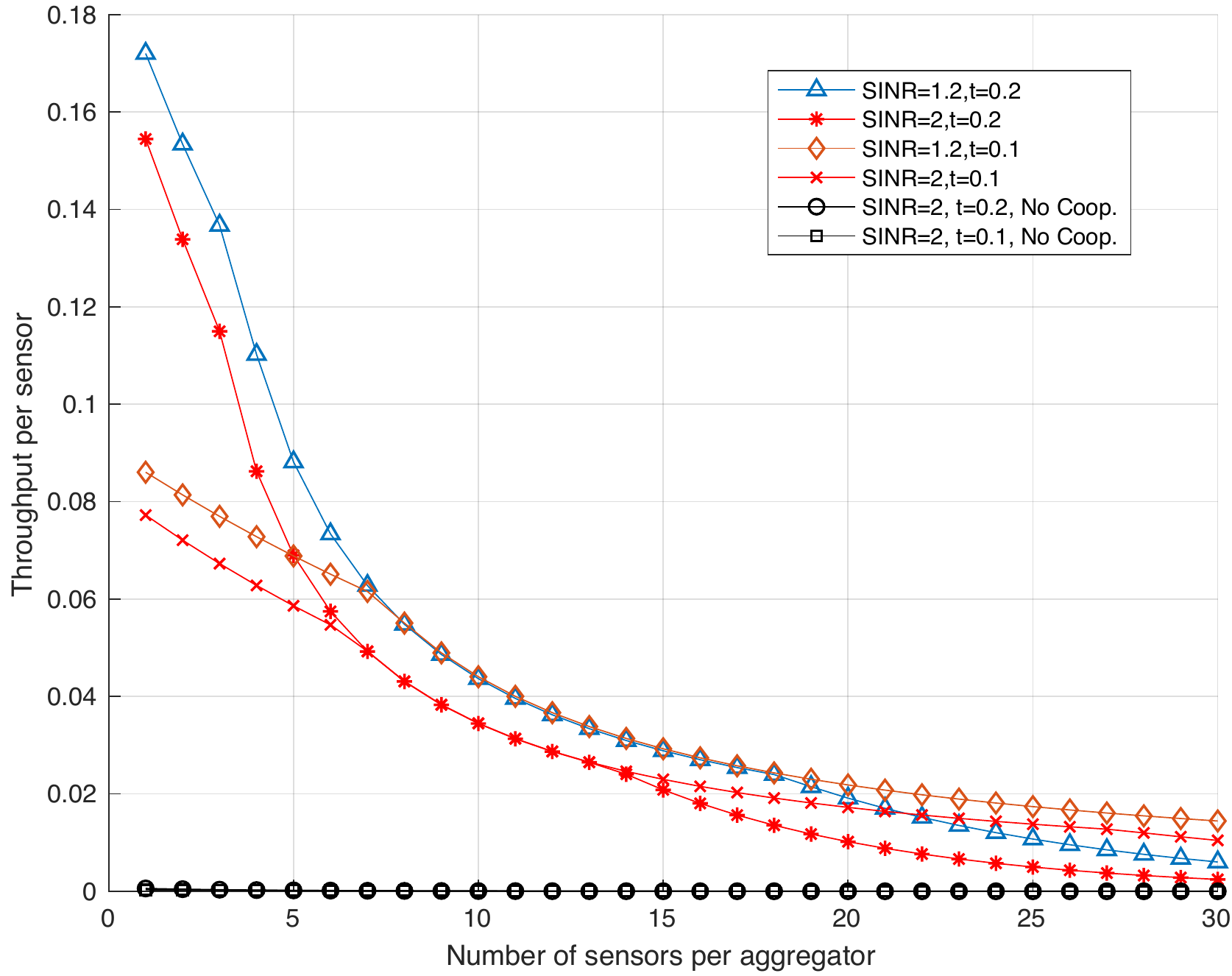}
		\label{fig:thrhigh}
	}
	\subfigure[Aggregate Throughput]{
		\includegraphics[scale=0.45]{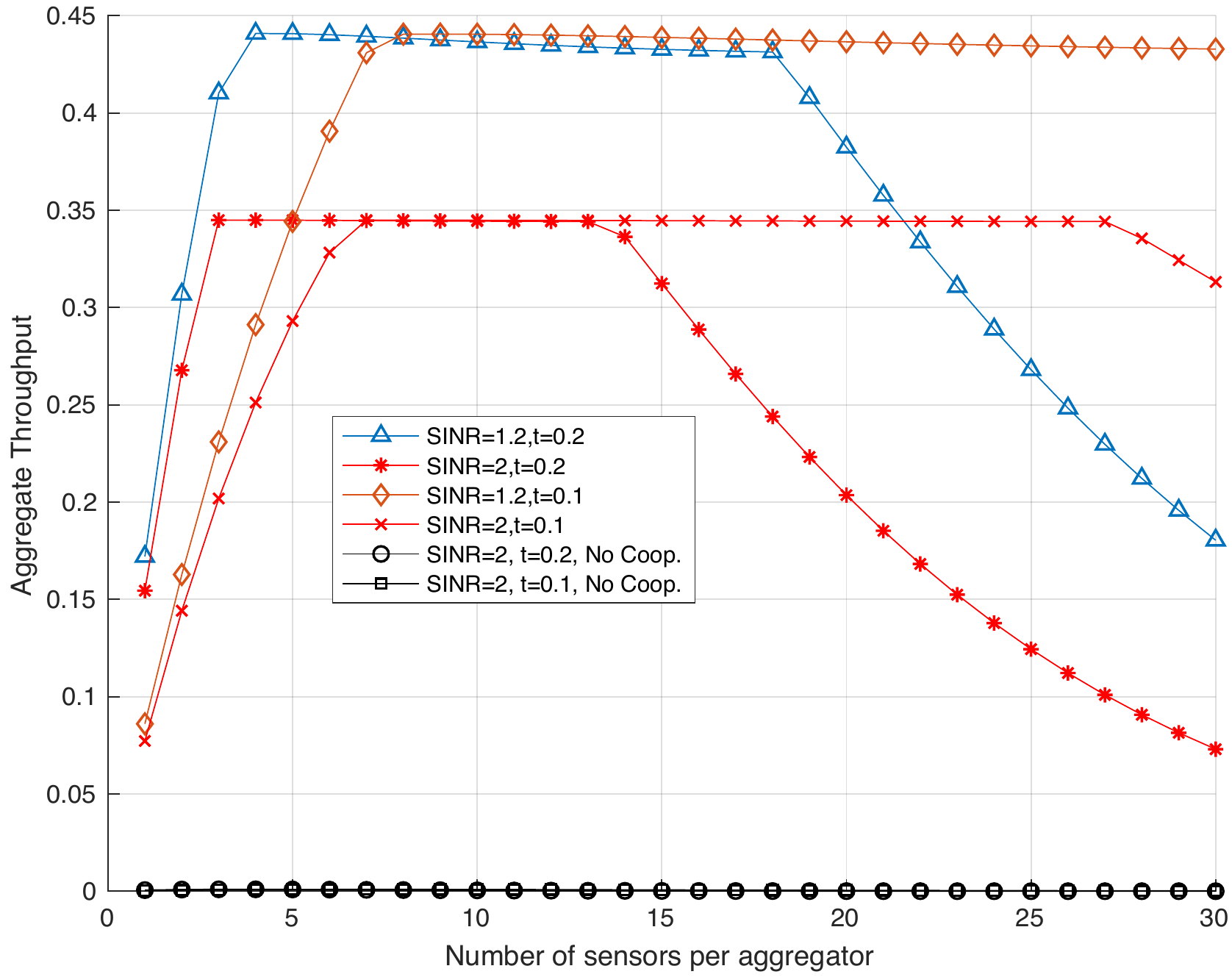}
		\label{fig:athrhigh}
	}\vspace{-0.1in}
	\caption{Throughput performance of the IoT network for SINR$ > 1$.} \label{fig:thhigh}
\end{figure}

In Fig. \ref{fig:rth}, we present the percentage of the traffic that is being relayed. We observe that in almost all cases the majority of the traffic is relayed.
Furthermore, in Fig. \ref{fig:rthrhigh} it becomes apparent that almost all the traffic that ends to the common destination is relayed traffic.

\begin{figure}[!ht]
	\centering
	\subfigure[SINR$ < 1$]{
		\includegraphics[scale=0.45]{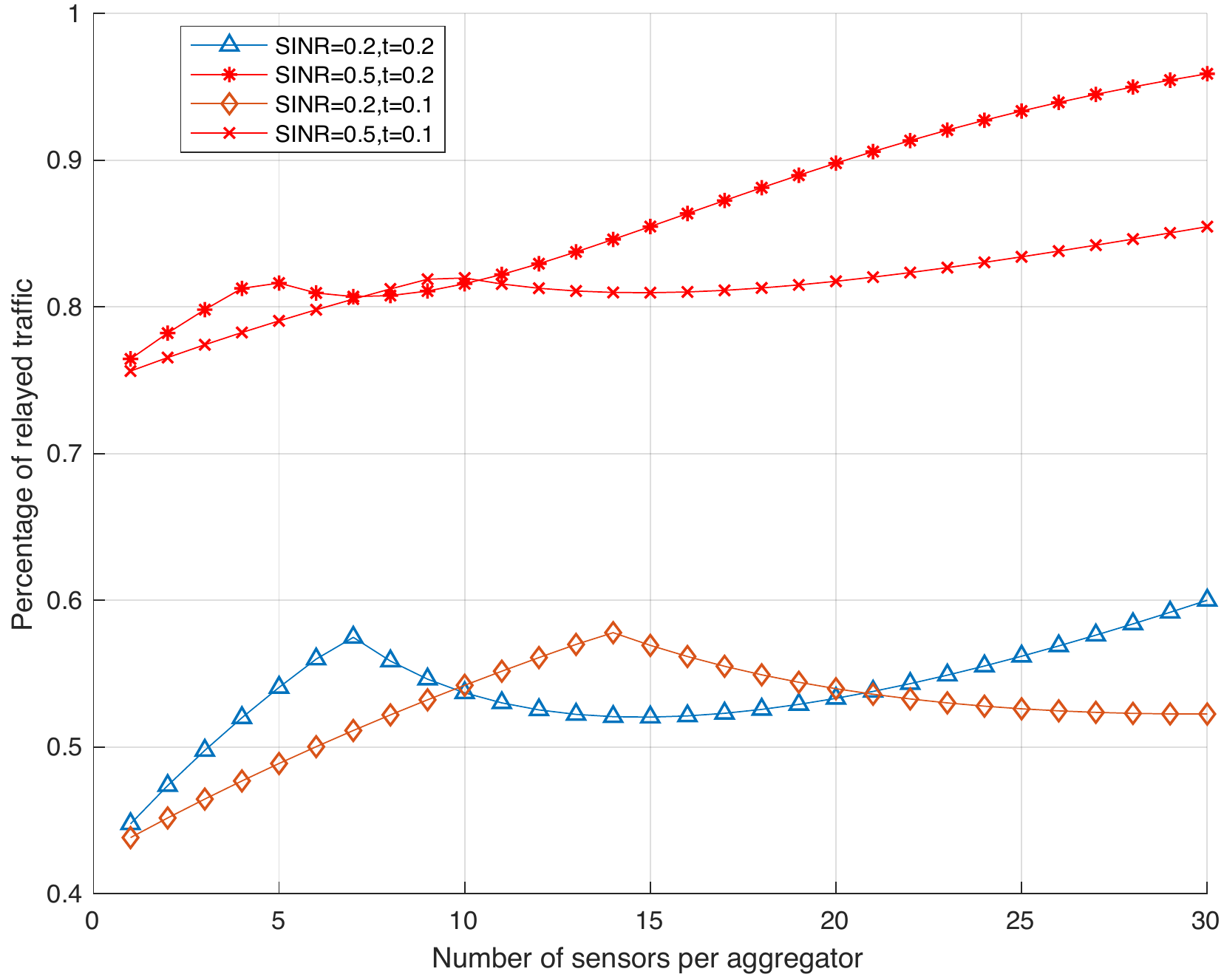}
		\label{fig:rthrlow}
	}
	\subfigure[SINR$ > 1$]{
		\includegraphics[scale=0.45]{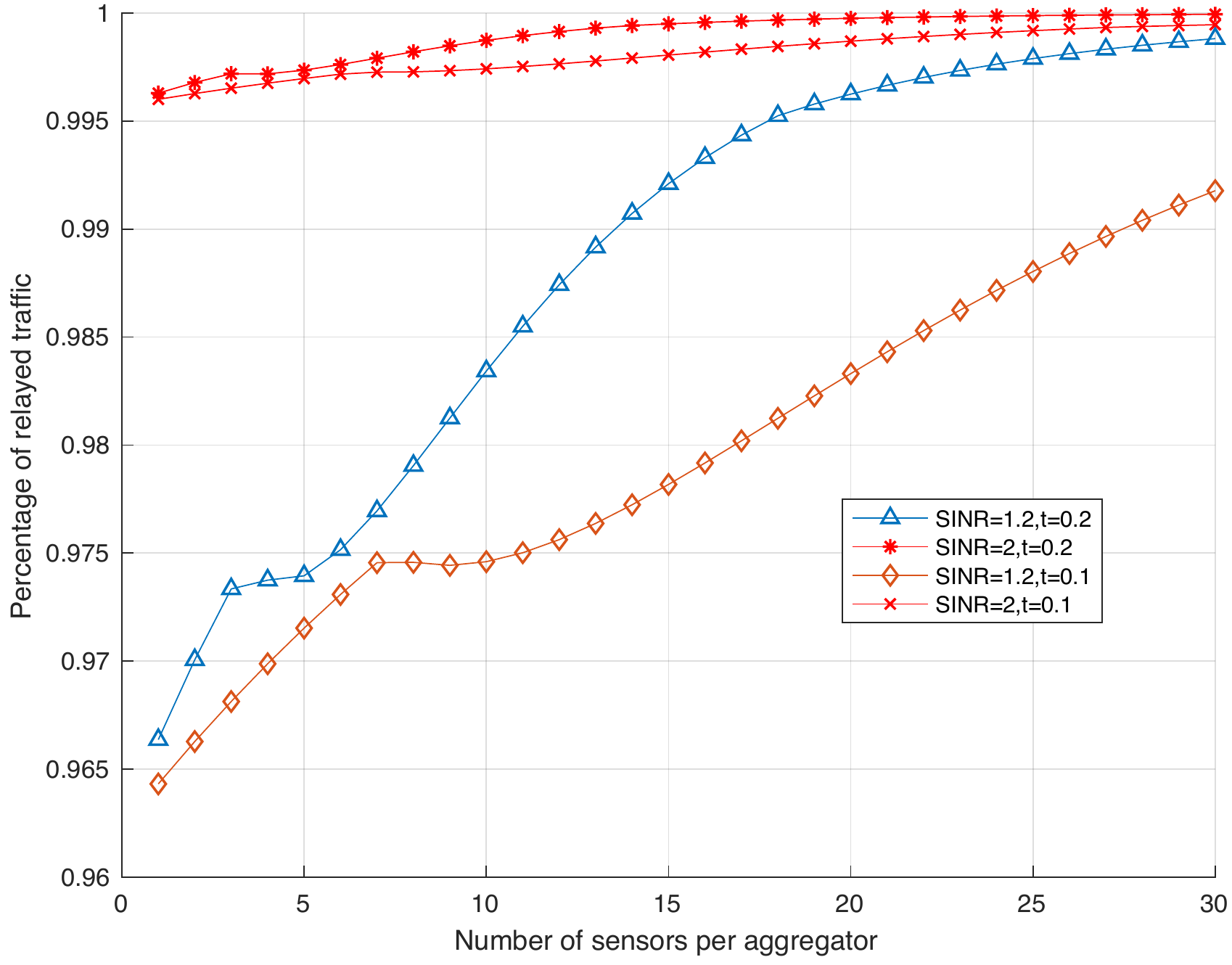}
		\label{fig:rthrhigh}
	}\vspace{-0.1in}
	\caption{Percentage of relayed traffic per sensor.} \label{fig:rth}
\end{figure}

Table \ref{tab:stable_sensors} provides what is the number of sensors per aggregator that can cause its queue stable or unstable.
As expected, when SINR$<1$, we have a clear transition from a stable queue to unstable queue as the number of sensors increase. However, the more interesting case is when SINR$>1$. We observe that for low number of sensors, we start with a stable queue, then the queue becomes unstable, and finally becomes stable again. The reason behind this is that initially, the incoming traffic increases up to a point the queue becomes unstable. Then due to the increased interference, the incoming traffic is reduced significantly thus, the queues are stable.
The shape of the curves presented in Fig. \ref{fig:rth} can be explained by Table \ref{tab:stable_sensors}. These results provide useful guidelines on the number of sensors that can be supported by an aggregator in order to achieve a required network performance regarding throughput.

\begin{table}[h]
\begin{center}
	\renewcommand{\arraystretch}{0.7}
	\caption{State of the queue per aggregator and the number of sensors.}
\begin{tabular}{|c|c|c|}
\hline
 & Stable & Unstable \\\hline
$SINR=0.2, t=0.2$ & $\{1,...,6\}$ & $\{7,...,30\}$ \\\hline
$SINR=0.5, t=0.2$ & $\{1,...,4\}$ & $\{5,...,30\}$ \\\hline
$SINR=0.2, t=0.1$ & $\{1,...,13\}$ & $\{14,...,30\}$ \\\hline
$SINR=0.5, t=0.1$ & $\{1,...,9\}$ & $\{10,...,30\}$ \\\hline \hline
$SINR=1.2, t=0.2$ & $\{1,...,3\},\{19,...,30\} $ & $\{4,...,18\}$ \\\hline
$SINR=2, t=0.2$ & $\{1,2\}, \{14,...,30\}$ & $\{3,...,13\}$ \\\hline
$SINR=1.2, t=0.1$ & $\{1,...,7\}$ & $\{8,...,30\}$ \\\hline
$SINR=2, t=0.1$ & $\{1,...,6\}, \{28,...,30\}$ & $\{7,...,27\}$ \\\hline
\end{tabular}
\label{tab:stable_sensors}
\end{center}
\end{table}

\subsection{Delay Performance}
In Fig. \ref{fig:delay}, the lower and the upper bounds for the average delay per packet versus the arrival rate at the aggregator are depicted for the four cases of SINR threshold, $0.2$, $0,5$, $1.2$ and $2$ respectively. The values of arrival rates are obtained by the stability conditions and we can connect the value of arrival rate with the number of users and their transmission probability through expression \eqref{eq:lambdaju}. Furthermore, we observe that the average delay increases with the increase of the SINR threshold. It is also important to notice that the upper and the lower bounds are very close, as an indication that the obtained bounds are tight.

\begin{figure}[!ht]
\centering
\includegraphics[scale=0.5]{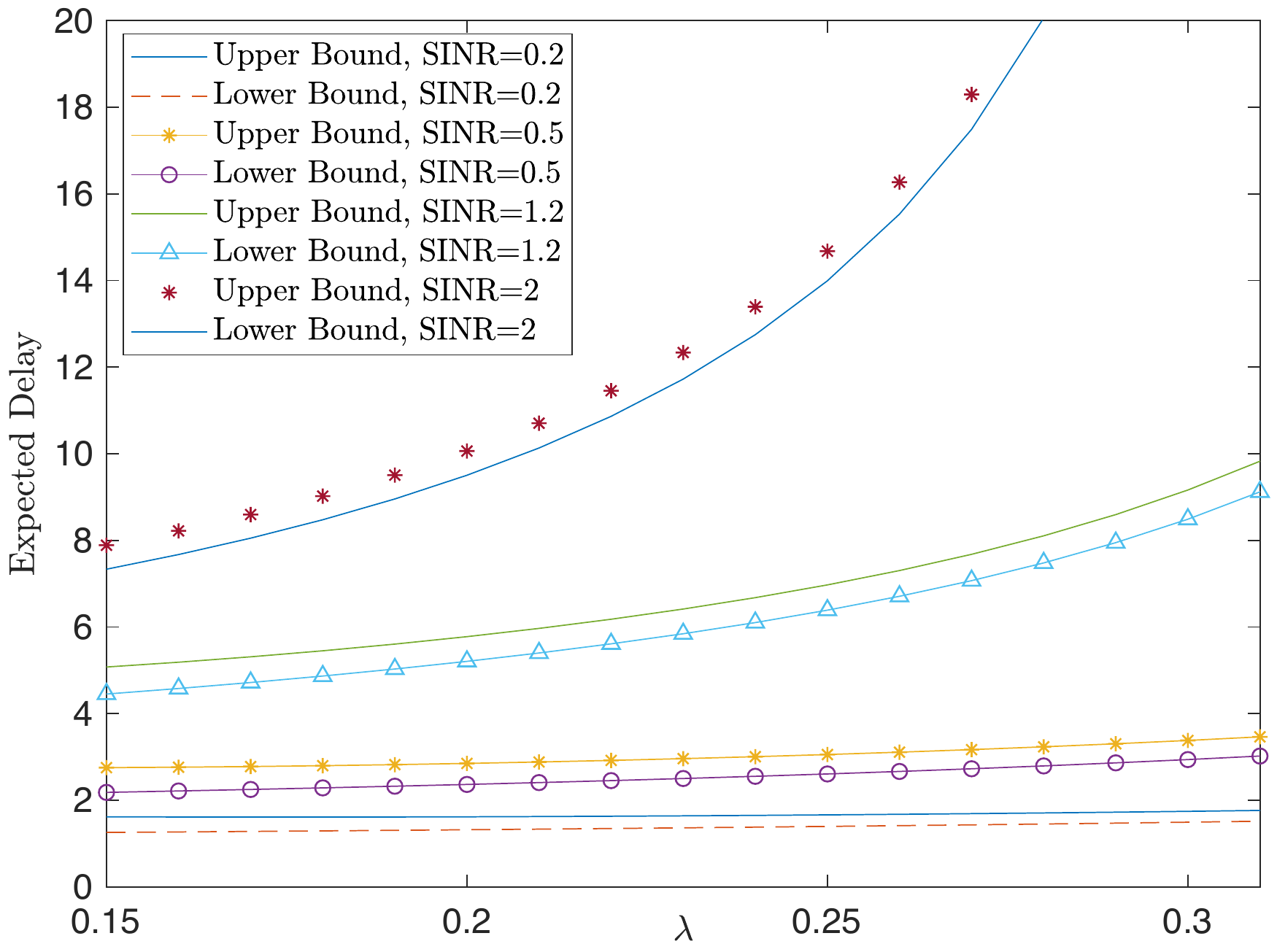}\vspace{-0.1in}
\caption{Average delay performance versus average arrival rate for various values of the SINR threshold for the case $\alpha=0.8$.}
\label{fig:delay}
\end{figure}

\section{Conclusions and Future Directions} \label{sec:conclusions}

In this work, we considered a random access IoT wireless network assisted by two aggregators providing support for data collection by applying network level cooperation. We characterized the throughput performance of the IoT network and obtained the stability conditions for the queues at the aggregators, which guarantee finite queueing delay. Furthermore, by applying the theory of boundary value problems we provided a detailed analysis for the delay. We showed that the benefits for the deployment of aggregators are more profound for higher SINR threshold and when the MPR capability of the network is weak. Our analytical results provide useful design guidelines for deploying aggregators in random access wireless IoT networks.

The suggested framework in this work can be extended to capture the case of cognitive aggregators that are adapting according to the incoming traffic from the sensors and also the channel conditions.
Considering adaptation of the operation of the aggregator according to the incoming traffic, can provide also performance improvements related to power consumption and energy efficiency, which is also another important factor for the design of IoT networks. In addition, it will be of further interest to consider delay-aware operation protocols for the IoT network and using successive interference cancellation to mitigate interference among the aggregators at the destination. The case of multiple aggregators in the same coverage area is an interesting future direction, our approach can also be applied but with some modifications.

\appendices

\section{Proof of Theorem \ref{thm1}} \label{sec:stability_proof}
To determine the stability region of the aggregators in our network, we apply the stochastic dominance technique \cite{Rao_TIT1988}. More specifically, we construct hypothetical dominant systems, in which an aggregator transmits dummy packets for the packet queue that is empty, while for the non-empty queue it transmits according to its traffic. Under this approach, we consider the $\mathcal{D}_{1}$, and $\mathcal{D}_{2}$-dominant systems. In the $\mathcal{D}_{k}$ dominant system, whenever the queue of aggregator $k$, $k=1,2$ empties, it transmits a dummy packet instead.

Thus, in $\mathcal{D}_{1}$, the first aggregator never empties, and hence, the second aggregator sees a constant service rate, while the service rate of aggregator $1$ depends on the state of aggregator $2$, i.e., empty or not. We proceed with the queue at aggregator $1$, its service rate is given by
\begin{equation} \label{eq:mu1}
\begin{array}{l}
\mu_1=Pr(N_2 \neq 0) \left[\alpha_1\bar{\alpha}_{2}p_{R_{1},\{R_{1}\}}^{D}+\alpha_1 \alpha_2 p_{R_1/R_1,R_2}^{D}\right] + Pr(N_2 = 0) \alpha_{1} p_{R_{1},\{R_{1}\}}^{D}.
\end{array}
\end{equation}
The service rate of the second aggregator is given by
\begin{equation} 
\begin{array}{c}
\mu_2=\alpha_2\bar{\alpha}_{1}p_{R_{2},\{R_{2}\}}^{D}+\alpha_1 \alpha_2 p_{R_2/R_1,R_2}^{D}.
\end{array}
\end{equation}
By applying Loyne's criterion, the second node is stable if and only if the average arrival rate is less that the average service rate, $\lambda_{2} < \alpha_2\bar{\alpha}_{1}p_{R_{2},\{R_{2}\}}^{D}+\alpha_1 \alpha_2 p_{R_2/R_1,R_2}^{D}$. We can obtain the probability that the second node is empty and is given by $Pr(N_2 = 0)=1-\frac{\lambda_{2}}{\mu_2}$. After replacing $Pr(N_2 = 0)$ into \eqref{eq:mu1}, and applying Loynes criterion we can obtain the stability condition for the first node. Then, we have the stability region $\mathcal{R}_1$ given by \eqref{eq:stabCond} for $i=1$. 
Similarly, we can obtain the stability region for the second dominant system $\mathcal{R}_2$. For a detailed treatment of dominant systems please refer to \cite{Rao_TIT1988}.

An important observation made in \cite{Rao_TIT1988} is that the stability conditions obtained by the stochastic dominance technique are not only sufficient but also necessary for the stability of the original system. The \emph{indistinguishability} argument \cite{Rao_TIT1988} applies here as well. Based on the construction of the dominant system, we can see that the queue sizes in the dominant system are always greater than those in the original system, provided they are both initialized to the same value and the arrivals are identical in both systems. Therefore, given $\lambda_{2}<\mu_{2}$, if for some $\lambda_{1}$, the queue at the first user is stable in the dominant system, then the corresponding queue in the original system must be stable. Conversely, if for some $\lambda_{1}$ in the dominant system, the queue at the first node saturates, then it will not transmit dummy packets, and as long as the first user has a packet to transmit, the behavior of the dominant system is identical to that of the original system since dummy packet transmissions are eliminated as we approach the stability boundary. Therefore, the original and the dominant system are indistinguishable at the boundary points.

\section{Proof of Lemma \ref{LEM}}\label{a0}
It is readily seen that $R(x,y)=\frac{xy-\Psi(x,y)}{xy}$, where $\Psi(x,y)=L(x,y)[xy+y(1-x)\alpha_{1}\widehat{\alpha}_{2}+x(1-y)\alpha_{2}\widehat{\alpha}_{1}]$, where for $|x|\leq1$, $|y|\leq1$, $\Psi(x,y)$ is a generating function of a proper probability distribution. Now, for $|y|=1$, $y\neq1$ and $|x|=1$ it is clear that $|\Psi(x,y)|<1=|xy|$. Thus, a direct application of Rouch\'e's theorem states that, $xy-\Psi(x,y)$ has exactly one zero inside the unit circle. Therefore, $R(x,y)=0$ has exactly one root $x=X_{0}(y)$, such that $|x|<1$. For $y=1$, $R(x,1)=0$ implies
%\begin{displaymath}
%\begin{array}{c}
$(1-x)\left(\lambda_{1}+\lambda_{1}\frac{\alpha_{1}\widehat{\alpha}_{2}(1-x)}{x}-\frac{\alpha_{1}\widehat{\alpha}_{2}}{x} \right)=0.$
%\end{array}
%\end{displaymath}
Therefore, for $y=1$, and since $\lambda_{1}<\alpha_{1}\widehat{\alpha}_{2}$, the only root of $R(x,1)=0$ for $|x|\leq1$, is $x=1$.\hfill$\square$
\section{Proof of Lemma \ref{SQ}}\label{a1}
For $y\in[y_{1},y_{2}]$, $D_{y}(y)$ is negative, so $X_{0}(y)$, $X_{1}(y)$ are complex conjugates. Therefore, $|X(y)|^{2}=\frac{\widehat{c}(y)}{\widehat{a}(y)}=g(y)$. Clearly, $g(y)$ is an increasing function for $y\in[0,1]$ and thus, $|X(y)|^{2}\leq g(y_{2})=\beta_{0}$. Using simple algebraic considerations we can prove that, $X_{0}(y_{1}):=\beta_{1}=-g(y_{1})$ is the extreme left point of $\mathcal{M}$. Finally, $\zeta(\delta)$ is derived by solving $Re(X(y))=-\widehat{b}(y)/2\widehat{a}(y)$ for $y$ with $\delta = Re(X(y))$, and taking the solution such that $y\in[0,1]$.\hfill$\square$
\section{Construction of the conformal mappings}\label{conf}
Our goal in order to to obtain expressions for the basic performance metrics is to construct the mapping $\gamma_{0}(z)$. To proceed, we need a representation of $\mathcal{M}$ in polar coordinates, i.e., $\mathcal{M}=\{x:x=\rho(\phi)\exp(i\phi),\phi\in[0,2\pi]\}.$

In the following, we summarize the basic steps: Since $0\in G_{\mathcal{M}}$, for each $x\in\mathcal{M}$, a relation between its absolute value and its real part is given by $|x|^{2}=m(Re(x))$ (see Lemma \ref{SQ}). Given the angle $\phi$ of some point on $\mathcal{M}$, the real part of this point, say $\delta(\phi)$, is the zero of $\delta-\cos(\phi)\sqrt{m(\delta)}$, $\phi\in[0,2\pi].$ Since $\mathcal{M}$ is a smooth, egg-shaped contour, the solution is unique. Clearly, $\rho(\phi)=\frac{\delta(\phi)}{\cos(\phi)}$, and the parametrization of $\mathcal{M}$ in polar coordinates is fully specified.

 Then, the mapping from $z\in G_{\mathcal{C}}$ to $x\in G_{\mathcal{M}}$, where $z = e^{i\phi}$ and $x= \rho(\psi(\phi))e^{i\psi(\phi)}$, satisfying $\gamma_{0}(0)=0$, $\gamma_{0}(z)=\overline{\gamma_{0}(\bar{z})}$ is uniquely determined by (see \cite{coh}, Section I.4.4),
\begin{equation}
\begin{array}{rl}
\gamma_{0}(z)=&z\exp[\frac{1}{2\pi}\int_{0}^{2\pi}\log\{\rho(\psi(\omega))\}\frac{e^{i\omega}+z}{e^{i\omega}-z}d\omega],\,|z|<1,\\
\psi(\phi)=&\phi-\int_{0}^{2\pi}\log\{\rho(\psi(\omega))\}\cot(\frac{\omega-\phi}{2})d\omega,\,0\leq\phi\leq 2\pi,
\end{array}
\label{zx}
\end{equation}
i.e., $\psi(.)$ is uniquely determined as the solution of a Theodorsen integral equation with $\psi(\phi)=2\pi-\psi(2\pi-\phi)$. This integral equation has to be solved numerically by an iterative procedure. For the numerical evaluation of the integrals we split the interval $[0,2\pi]$ into $M$ parts of length $2\pi/M$, by taking $M$ points $\phi_{k}=\frac{2k\pi}{M}$, $k=0,1,...,M-1$. For the $M$ points given by their angles $\left\{\phi_{0},...,\phi_{M-1}\right\}$ we should solve the second in (\ref{zx}) to obtain the corresponding points $\left\{\psi(\phi_{0}),...,\psi(\phi_{M-1})\right\}$, iteratively from,
\begin{equation}
%\begin{array}{rl}
\psi_{0}(\phi_{k})=\phi_{k},
\psi_{n+1}(\phi_{k})=\phi_{k}-\frac{1}{2\pi}\int_{0}^{2\pi}\log\left\{\frac{\delta(\psi_{n}(\omega))}{\cos(\psi_{n}(\omega))}\right\}\cot\left(\frac{\omega-\phi_{k}}{2}\right)d\omega,
%\end{array}
\label{uip}
\end{equation}
where $\lim_{n\to\infty}\psi_{n+1}(\phi)=\psi(\phi)$, and $\delta(\psi_{n}(\omega))$ is determined by, $\delta(\psi_{n}(\omega))=cos(\psi_{n}(\omega))\sqrt{m(\delta(\psi_{n}(\omega)))}$, using the Newton-Raphson root finding method. For each step, the integral in (\ref{uip}) is
numerically determined by again using the trapezium rule with $M$ parts of equal length
$2\pi/M$. For the iteration, we have used the following stopping criterion $\max_{k\in\left\{0,1,...,M-1\right\}}\left|\psi_{n+1}(\phi_{k})-\psi_{n}(\phi_{k})\right|<10^{-6}$

After obtaining $\psi(\phi)$ numerically, the values of the conformal mapping $\gamma_{0}(z)$, $\left|z\right|\leq 1$, can be calculated by applying the Plemelj-Sokhotski formula to the first in (\ref{zx}) for $0\leq\phi\leq 2\pi$,
\begin{displaymath}
\begin{array}{rl}
\gamma_{0}(e^{i\phi})=\frac{e^{i\psi(\phi)}\delta(\psi(\phi))}{\cos(\psi(\phi))}=\delta(\psi(\phi))[1+i \tan(\psi(\phi))].
\end{array}
\end{displaymath}
We can find $\gamma(1)$, $\gamma^{\prime}(1)$ by applying the Newton's method and solving $\gamma_{0}(z_{0})=1$, in $[0,1]$, i.e., $z_{0}$ is the zero in $[0,1]$ of $\gamma_{0}(z)=1$. Then, $\gamma(1)=z_{0}$. Moreover, using the first in (\ref{zx})
\begin{equation}
\begin{array}{rl}
\gamma^{\prime}(1)=&(\gamma_{0}^{\prime}(z_{0}))^{-1}=\left\lbrace\frac{1}{\gamma(1)}+\frac{1}{2\pi i}\int_{0}^{2\pi}\frac{\log\{\rho(\psi(\omega))\}2e^{i\omega}}{(e^{i\omega}-\gamma(1))^{2}}d\omega\right\rbrace^{-1},
\end{array}
\label{cv}
\end{equation}
which can be obtained numerically by using the trapezoidal rule for the integral on the right-hand side of (\ref{cv}).

\bibliographystyle{IEEEtran}
\bibliography{bibliography}
\end{document}